%% file: main-temporal-walk-long.tex
\documentclass[a4paper,USenglish,autoref]{lipics-v2021}

\hideLIPIcs  

\title{Kernelizing Temporal Exploration Problems} 

\author{Emmanuel Arrighi}{University of Bergen, Norwayi \and \url{http://emmanuel.arrighi.eu}}{emmanuel.arrighi@gmail.com}{https://orcid.org/0000-0002-0326-1893}{}

\author{Fedor V. Fomin}{University of Bergen, Norway}{Fedor.Fomin@uib.no}{}{}

\author{Petr Golovach}{University of Bergen, Norway}{Petr.Golovach@ii.uib.no}{}{}

\author{Petra Wolf}{University of Bergen, Norway \and \url{https://www.wolfp.net/}}
{mail@wolfp.net}
{https://orcid.org/0000-0003-3097-3906}
{}

\authorrunning{E. Arrighi, F. Fomin, P. Golovach and P. Wolf} 

\Copyright{Emmanuel Arrighi Fedor V. Fomin, Petr Golovach and Petra Wolf} 

\ccsdesc[500]{Theory of computation~Parameterized complexity and exact algorithms}
\ccsdesc[500]{Theory of computation~Problems, reductions and completeness}

\keywords{Temporal graph, temporal exploration, computational complexity, kernel} 

\category{} 

\relatedversion{} 

\acknowledgements{}

\nolinenumbers 

\EventEditors{John Q. Open and Joan R. Access}
\EventNoEds{2}
\EventLongTitle{42nd Conference on Very Important Topics (CVIT 2016)}
\EventShortTitle{CVIT 2016}
\EventAcronym{CVIT}
\EventYear{2016}
\EventDate{December 24--27, 2016}
\EventLocation{Little Whinging, United Kingdom}
\EventLogo{}
\SeriesVolume{42}
\ArticleNo{23}

\usepackage{xspace}
\newcommand{\FPT}{\textsf{FPT}\xspace}
\newcommand{\XP}{\textsf{XP}\xspace}
\newcommand{\NP}{\textsf{NP}\xspace}
\newcommand{\PTime}{\textsf{P}\xspace}
\newcommand{\coNP}{\textsf{coNP}\xspace}
\newcommand{\poly}{\textsf{poly}\xspace}
\newcommand{\W}{\textsf{W}}

\usepackage{tikz}
\usetikzlibrary{arrows,calc}
\usetikzlibrary{shapes,automata}
\usetikzlibrary{decorations.pathreplacing,calligraphy}

\usepackage[disable]{todonotes}
\usepackage{booktabs}
\usepackage{amsmath}
\usepackage[appendix=inline]{apxproof}
 \usepackage{tcolorbox}

\DeclareMathOperator{\operatorClassNP}{{\sf NP}}
\newcommand{\classNP}{\ensuremath{\operatorClassNP}}
\DeclareMathOperator{\operatorClassCoNP}{{\sf coNP}}
\newcommand{\classCoNP}{\ensuremath{\operatorClassCoNP}}
\DeclareMathOperator{\operatorClassFPT}{{\sf FPT}\xspace}
\newcommand{\classFPT}{\ensuremath{\operatorClassFPT}\xspace}
\DeclareMathOperator{\operatorClassW}{{\sf W}}
\newcommand{\classW}[1]{\ensuremath{\operatorClassW[#1]}}

\DeclareMathOperator{\operatorClassXP}{{\sf X}P\xspace}
\newcommand{\classXP}{\ensuremath{\operatorClassXP}\xspace}

\newcommand{\Oh}{\mathcal{O}}

\newcommand{\m}{\mathfrak{m}\xspace}

\newcommand{\G}{\mathcal{G}\xspace}
\newcommand{\T}{\mathcal{T}\xspace}

\newcommand{\GL}{\mathcal{G}= (G_1,\ldots,G_L)\xspace}

\newcommand{\WEX}{\textsc{Weighted $k$-arb NS-TEXP}\xspace}

\newtheorem{redrule}{Reduction Rule}

\begin{document}

\maketitle

\begin{abstract}
We study the kernelization of exploration problems on temporal graphs. A temporal graph consists of a finite sequence of snapshot graphs $\mathcal{G}=(G_1, G_2, \dots, G_L)$ that share a common vertex set but might have different edge sets. 
%
%
%
 The non-strict temporal exploration problem (\textsc{NS-TEXP} for short) introduced by Erlebach and Spooner, asks if a single agent can visit all vertices of a given temporal graph where the edges traversed by the agent are present in non-strict monotonous time steps, i.e., the agent can move along the edges of a snapshot graph with infinite speed. The exploration must at the latest be completed in the last snapshot graph. The optimization variant of this problem is the \textsc{$k$-arb NS-TEXP} problem, where the agent's task is to visit at least $k$  vertices of the temporal graph.
We show that under standard computational complexity assumptions, neither of the problems \textsc{NS-TEXP} nor \textsc{$k$-arb NS-TEXP} allow for polynomial kernels in the standard parameters: number of vertices $n$, lifetime $L$, number of vertices to visit $k$, and maximal number of connected components per time step $\gamma$; as well as in the combined parameters $L+k$, $L + \gamma$, and $k+\gamma$. On the way to establishing these lower bounds, we answer a couple of questions left open by  Erlebach and Spooner.

We also initiate the study of structural kernelization by identifying a new parameter of a temporal graph 
 $p(\mathcal{G}) = \sum_{i=1}^{L} (|E(G_i)|) - |V(G)| +1$. Informally, this parameter measures how dynamic the temporal graph is. Our main algorithmic result is the construction of a polynomial (in $p(\mathcal{G})$) kernel for the more general \textsc{Weighted $k$-arb NS-TEXP}  problem, where weights are assigned to the vertices and the task is to find a temporal walk of weight at least $k$. 
 
%


\end{abstract}
\section{Introduction}
We investigate the kernelization of exploration and connectivity tasks in temporal networks. While kernelization, and in particular, structural kernelization, appears to be a successful approach for addressing many optimization problems on \emph{static} graphs, its applications in temporal graphs are less impressive, to say the least. A reasonable explanation for this  (we will provide some evidence of that later)  is that most of the structural parameters of statical graphs, like treewidth, size of a feedback vertex set, or the vertex cover number, do not seem to be useful when it comes to dynamic settings. This brings us to the following question. 

\begin{tcolorbox}[colback=green!5!white,colframe=blue!40!black]
   What structure of dynamic graphs could be ``helpful'' for kernelization algorithms?
\end{tcolorbox}
    
We propose a new structural parameter of a temporal graph estimating how ``dynamic'' is the temporal graph.  Our main result is an algorithm for the exploration problem on a temporal graph that produces a kernel whose size is polynomial in the new parameter. Before proceeding with the formal statement of the problem and our results, we provide a short overview of temporal graphs and kernelization.

\medskip\noindent\textbf{Temporal  exploration.}
Many networks considered nowadays show an inherently dynamic behavior, for instance connectivity in mobile ad-hoc networks, relationships in a social network, or accessibility of services in the internet. Classical graphs do not suffice to model those dynamic behaviors, which led to an intensive study of so-called \emph{temporal graphs}. In general, temporal graphs are graphs that change over time. There are several models in the literature that consider different types of dynamic behaviors and encodings of the temporal graphs. Here, we consider a temporal graph $\mathcal{G}$ to consist of a sequence of \emph{snapshot} graphs $G_1, G_2, \dots, G_L$ that share a common set of vertices $V$ but might have different edge sets $E_1, E_2,\dots , E_L$. We call the graph $G = (V, \bigcup_{i=1}^L E_i)$ the \emph{underlying graph} of $\mathcal{G}$. We can think of the temporal graph $\mathcal{G}$ as the graph $G$ where only a subset of edges are present at a certain time step. The number of snapshots $L$ is commonly referred to as the \emph{lifetime} of $\mathcal{G}$. Due to its relevance in modeling dynamic systems, a huge variety of classical graph problems have been generalized to temporal graphs. We refer to \cite{DBLP:journals/paapp/CasteigtsFQS12,DBLP:journals/jcss/KempeKK02,DBLP:journals/im/Michail16,RozenshteinG19} for an introduction to temporal graphs and their associated combinatorial problems.

Graph exploration is a fundamental problem in the realm of graph theory and has been extensively studied since its introduction by Shannon in 1951~\cite{shannon1993presentation}. The question if a given graph can be explored by a single agent was generalized to temporal graphs by Michail and Spirakis \cite{DBLP:journals/tcs/MichailS16}. They showed that deciding whether a single agent can visit all vertices of a temporal graph within a given number of time steps while crossing only one edge per time step is an \NP-hard problem. This hardness motivated the study of the parameterized complexity of the temporal exploration problem performed by Erlebach and Spooner in~\cite{DBLP:conf/sand/ErlebachS22}. 

While the temporal exploration problem, \textsc{TEXP} for short, allows the agent to cross only one edge per time step, in some scenarios, the network only changes slowly, way slower than the speed of an agent. Then, it is natural to assume that the agent can travel with infinite speed and explore the whole connected component in which he is currently located in a single time step. Such scenarios arise for instance in delay-tolerant networks~\cite{DBLP:journals/jnca/CRMS16}. Erlebach and Spooner generalized \textsc{TEXP} to allow the agent to travel with infinite speed in~\cite{DBLP:conf/sirocco/ErlebachS20} and called this problem the \emph{non-strict temporal exploration problem}, \textsc{NS-TEXP} for short.
In \textsc{NS-TEXP} the task is to identify whether an agent can visit all vertices of the graph starting from an initial vertex $v$ according to the following procedure. At step~$1$, the agent visits all vertices of the connected component $C_1$ of the snapshot graph $G_1$ containing~$v$. At step $i\in\{2,\dots, L\}$, the agent selects to explore the vertices of one of the components $C_i$ of $G_i$ subject that this component intersects the component explored by the agent at step $(i-1)$.

 Shifting to an infinite exploration speed changes the nature of the temporal exploration problem drastically, as now the exploration time can be significantly shorter than the number of vertices. While for classical graphs, the question if a graph can be explored with infinite speed simply reduces to the question of connectivity; for temporal graphs,  \textsc{NS-TEXP} is \NP-complete~\cite{DBLP:conf/sirocco/ErlebachS20}. 
The parameterized complexity of \textsc{NS-TEXP} was studied by Erlebach and Spooner in~\cite{DBLP:conf/sand/ErlebachS22} where \FPT algorithms for the parameter \emph{lifetime} $L$ of the temporal graph (equivalently number of time steps during which the temporal exploration must be performed) was shown. They also studied an optimization variant of \textsc{NS-TEXP}, where not all but at least  $k$ vertices must be visited during the lifetime of the graph. This variant is called \textsc{$k$-arb NS-TEXP}. An \FPT algorithm with parameter $k$ solving \textsc{$k$-arb NS-TEXP} was obtained in~\cite{DBLP:conf/sand/ErlebachS22}. In the long version of this study~\cite{ErlebachS22Journal}, it was stated as an open question whether \textsc{NS-TEXP} is in \FPT or at least in \XP for the parameter maximal number of connected components in a time step, $\gamma$, and, whether \textsc{k-arb NS-TEXP} is in \FPT for parameter $L$. These open problems were also stated at the 2022 ICALP satellite workshop \emph{Algorithmic Aspects of Temporal Graphs V}.
In this work, we will answer both of these questions negatively.

\medskip
\noindent
\textbf{Kernelization.} Informally,  {kernelization} is a {{preprocessing algorithm}} that consecutively applies various data reduction rules in order to shrink the instance size of a parameterized problem as much as possible. 
Kernelization is one of the major research domains of parameterized complexity and many important advances in the area are on kernelization. 
These advances include general algorithmic findings on problems admitting kernels of polynomial size and frameworks for ruling out polynomial kernels under certain complexity-theoretic assumptions. We refer to the book \cite{kernelizationbook19} for an overview of the area. 
A fruitful approach in kernelization (for static graphs)  is the study of the impact of various structural measurements (i.e., different than just the total input size or expected solution size) on the problem complexity. 
Such structural parameterizations, like the minimum size of a feedback vertex set,  of a vertex cover, or the treedepth, measure the non-triviality of the instance~\cite{BougeretS19,JansenB13,JansenK021,UhlmannW13}.

\medskip
\noindent
\textbf{Our contribution.} Our focus is on the possibility of kernelization for the problems \textsc{NS-TEXP} and \textsc{$k$-arb NS-TEXP}.  
In their work on the parameterized complexity of these problems, Erlebach and Spooner in~\cite{DBLP:conf/sand/ErlebachS22,ErlebachS22Journal} established fixed-parameter tractability for some combinations of these problems and ``standard'' parameterization like 
number of vertices~$n$, lifetime~$L$, number of vertices to visit~$k$, and $L+\gamma$, where $\gamma$ is the maximal number of connected components per time step, see \autoref{tab:overview}. 
As the first step, we rule out the existence of polynomial kernels for both problems when parameterized by each of these standard parameters. Moreover, our lower bounds hold even for ``combined'' parameters $L+k$, $L + \gamma$,  and $k+\gamma$.
On the way to establish our lower bounds for kernelization, we resolve two open problems from the parameterized study of \textsc{NS-TEMP} and \textsc{$k$-arb NS-TEXP} posed by Erlebach and Spooner in~\cite{DBLP:conf/sand/ErlebachS22,ErlebachS22Journal}.
 Namely, we show that \textsc{NS-TEXP} is \NP-complete for constant values of $\gamma \geq 5$ and that \textsc{$k$-arb NS-TEXP} is \W[1]-hard parameterized by $L$.

This motivates the study of structural kernelization of the exploration problems. While at the first glance, the most natural direction would be to explore the structure of the underlying graph of the temporal graph  $\mathcal{G}$, this direction does not seem to bring new algorithmic results. The reason is that our lower bounds on  \textsc{NS-TEXP} and \textsc{$k$-arb NS-TEXP} hold for temporal graphs with \emph{very} restricted underlying graphs. In particular, we show that the problems remain  \NP-hard even when the underlying graph is a tree with vertex cover number at most~$2$.  Thus a reasonable structural parameterization, in this case, should capture not only the ``static'' structure of the underlying graph but also the ``dynamics'' of the edges in  $\mathcal{G}$. With this in mind, we introduce the new parameter of a temporal graph $\mathcal{G} = (G_1, G_2, \dots, G_L)$, 
%
%
  \[p(\mathcal{G}) = \sum_{i=1}^{L} (|E(G_i)|) - |V(G)|+1.\]
  
   Note that if $\sum_{i=1}^{L} (|E(G_i)|) < |V(G)| - 1$, 
then the underlying graph is disconnected 
and hence the temporal graph $\mathcal{G}$ cannot be explored. The parameter $p(\mathcal{G})$ further bounds the number of edge-appearances of edges that appear multiple times and can thereby be seen as a measure of how dynamic the temporal graph is.
Our main result is a polynomial kernel for the more general problem \textsc{Weighted $k$-arb NS-TEXP} in the parameter $p=p(\mathcal{G})$. In \textsc{Weighted $k$-arb NS-TEXP}, the vertices contain weights and the task is to find a temporal graph that visits vertices with a total sum of weights of at least $k$ for some given integer $k$.
The obtained kernel is of size $\Oh(p^4)$ and contains a number of vertices that is linear in the parameter~$p$. 
Our results are summarized in \autoref{tab:overview}. 
\begin{table}	
	\centering
	\begin{tabular}{l|cc|cc}
		Param. & \multicolumn{2}{c}{\textsc{NS-TEXP}} & \multicolumn{2}{c}{\textsc{$k$-arb NS-TEXP}}\\
		& \FPT & Kernel 
		& \FPT & Kernel\\
		\toprule
		$p$ &\colorbox{red!20}{$2^{\Oh(p)}(nL)^{\Oh(1)}$} & \colorbox{red!20}{$\Oh(p^4)^\star$} & \colorbox{red!20}{$2^{\Oh(p)}(nL)^{\Oh(1)}$} & \colorbox{red!20}{$\Oh(p^4)^\star$}\\
		$n$ & $\Oh^*((2e)^n n^{\log{n}})$~\cite{ErlebachS22Journal} & \colorbox{red!20}{{no poly kernel}} & \FPT in $k$~\cite{ErlebachS22Journal} & \colorbox{red!20}{{no poly kernel}}\\
		$L$ & $\Oh^*(L(L!)^2)$~\cite{DBLP:conf/sand/ErlebachS22} & \colorbox{red!20}{{no poly kernel}} & \colorbox{red!20}{{\W[1]-hard}} & \colorbox{red!20}{{no poly kernel}}\\
		$k$ & - & - & $\Oh^*((2e)^k k^{\log{k}})$~\cite{ErlebachS22Journal} & \colorbox{red!20}{{no poly kernel}} \\
		$L+k$ & - & - & \FPT in  $k$~\cite{ErlebachS22Journal} & \colorbox{red!20}{{no poly kernel}}\\
		%
		$\gamma$ & in \PTime for $\leq 2$~\cite{ErlebachS22Journal},  & - & in \PTime for $\leq 2$~\cite{ErlebachS22Journal},  & - \\
		& \colorbox{red!20}{{\NP-hard for $\geq 5$}} & & \colorbox{red!20}{{\NP-hard for $\geq 5$}} & \\
$L+\gamma$ & \FPT in $L$~\cite{DBLP:conf/sand/ErlebachS22} & \colorbox{red!20}{{no poly kernel}} & \colorbox{red!20}{{$\Oh(\gamma^Ln^{\Oh(1)})$}} & \colorbox{red!20}{{no poly kernel}} \\
		& & \colorbox{red!20}{{for $\gamma\geq 6$}} & & \colorbox{red!20}{{for $\gamma\geq 6$}}\\
		$k+\gamma$ & - & - & \FPT in $k$~\cite{ErlebachS22Journal} & \colorbox{red!20}{no poly kernel}\\
	\end{tabular}
	\caption{Overview of the parameterized complexity of \textsc{NS-TEXP} and \textsc{$k$-arb NS-TEXP}. Our contribution is highlighted in red. The results stating that there is no polynomial kernel rely on the assumption that $\NP \not\subseteq \coNP/\poly$. Here, for a temporal graph $\mathcal{G}$, $n$ is the number of vertices, $L$ the lifetime, $\gamma$ is the maximal number of components per time step, and $p=p(\mathcal{G})$.
	 For entries marked with~$^\star$, we in fact get a generalized kernel for the more general variant \textsc{Weighted $k$-arb NS-TEXP}.}
	\label{tab:overview}
\end{table}

\noindent
\textbf{Further related work.} 
Michail and Spirakis~\cite{DBLP:journals/tcs/MichailS16} introduced the \textsc{TEXP} problem and showed that the problem is \NP-complete when no restrictions are placed on the input. They proposed considering the problem under the \emph{always-connected} assumption that requires that the temporal graph is connected in every time step. Erlebach et al.~\cite{DBLP:journals/jcss/Erlebach0K21} followed this proposition and showed that for always-connected temporal graphs, computing a foremost exploration schedule is \NP-hard to approximate with ratio $\Oh(n^{1-\epsilon})$, for every $\epsilon > 0$. They also showed that, under the always-connected assumption, subquadratic exploration schedules exist for temporal graphs whose underlying graph is planar, has bounded treewidth, or is a $2\times n$ grid; and linear  exploration schedules exist for cycles and cycles with one chord.  Alamouti~\cite{taghian2020exploring} showed that a cycle with $k$ chords can be explored in $\Oh(k^2 \cdot k! \cdot (2e)^k \cdot n)$ time steps. Adamson et al.~\cite{DBLP:conf/sand/AdamsonGMZ22} improved this bound to $\Oh(kn)$ time steps. They also improved the bounds on the worst-case exploration time for temporal graphs whose underlying graph is planar or has bounded treewidth.
Erlebach et al.~\cite{DBLP:conf/icalp/ErlebachKLSS19} showed that temporal graphs can be explored in $\Oh(n^{1.75})$ time steps if each snapshot graph admits a spanning-tree of bounded degree or if one is allowed to traverse two edges per step.
In~\cite{DBLP:journals/acta/ErlebachS22}, Erlebach and Spooner showed that always-connected, temporal graphs where in each snapshot graph at most $k$ edges of the underlying graph are \emph{not} present, can always be explored in $\Oh(kn\log{n})$ time steps. The same authors showed in~\cite{DBLP:conf/mfcs/ErlebachS18} that a temporal graph can be explored in $\Oh(\frac{n^2}{\log{n}})$ time steps, if each snapshot graph has bounded degree.
Bodlaender and van der Zanden~\cite{DBLP:journals/ipl/BodlaenderZ19} showed that the \textsc{TEXP} problem, when restricted to always-connected temporal graphs whose underlying graph has pathwidth at most 2, remains \NP-complete.
Erlebach and Spooner~\cite{DBLP:conf/sand/ErlebachS22} studied the \textsc{TEXP} problem from a parameterized perspective. Based on the color coding technique, they gave an \FPT algorithm parameterized by $k$ for the problems \textsc{$k$-arb TEXP} and the non-strict variant \textsc{$k$-arb NS-TEXP}. They also gave an \FPT algorithm parameterized by the lifetime $L$ of the temporal graph for the problems \textsc{TEXP} and \textsc{NS-TEXP}. 
In the respective long version~\cite{ErlebachS22Journal}, Erlebach and Spooner further studied the parameter maximal number of connected components per time step $\gamma$ and showed that \textsc{TEXP} is \NP-hard for $\gamma = 1$, but \textsc{NS-TEXP} is solvable in polynomial time for $\gamma \leq 2$.
The \textsc{NS-TEXP} problem was introduced and studied by Erlebach and Spooner~\cite{DBLP:conf/sirocco/ErlebachS20}. 
Among other things, they showed \NP-completeness of the general problem, as well as $\Oh(n^{1/2-\epsilon})$ and $\Oh(n^{1-\epsilon})$-inapproximability for computing a foremost exploration schedule under the assumption that the number of time steps required to move between any pair of vertices is bounded by 2, respectively 3.

\begin{toappendix}
Parameterized studies of a related problem where performed by Casteigts et al.~\cite{DBLP:journals/algorithmica/CasteigtsHMZ21} who studied the problem of finding temporal paths between a source and target that wait no longer than $\Delta$ consecutive time steps at any intermediate vertex.
\end{toappendix}
Bumpus and Meeks~\cite{bumpus2022edge}
 considered the parameterized complexity of a graph exploration problem that asks no longer to visit all vertices, but to traverse all \emph{edges} of the underlying graph exactly once. They observed that for natural structural parameters of the underlying graph, their considered exploration problem does not admit \FPT algorithms. Similarly, Kunz, Molter, and Zehavi~\cite{DBLP:journals/corr/abs-2301-10503} obtained several hardness results for structural parameters of the underlying graph when studying the parameterized complexity of the problem of finding temporally disjoint paths and walks, a problem introduced by Klobas et al.~\cite{DBLP:journals/aamas/KlobasMMNZ23}. Their obtained \W[1]-hardness in the parameter number of vertices, that holds even for instances where the underlying graph is a star, indicates that simply considering structural parameters of the underlying graph is not sufficient to obtain \FPT algorithms for temporal graph problems.

\section{Preliminaries}
\noindent
\textbf{Notations.}
We denote by $\mathbb{Z}$ the set of numbers, by $\mathbb{N}$ the set of natural numbers, by $\mathbb{N}_{>0}$ the set of positive numbers, and by $\mathbb{Q}$ the set of rational numbers.
Let $n$ be an positive integer, we denote with $[n]$ the set $\{1, 2, \dots, n\}$.
Given a vector $w = (w_1, w_2, \ldots, w_r) \in \mathbb{Q}^r$, we let
$\|w\|_{\infty} = \max_{i \in [r]} |w_i|$ and $\|w\|_1 = \sum_{i \in [r]} |w_i|$.

\medskip
\noindent
\textbf{Graphs.}
We consider a graph $G = (V, E)$ to be a static undirected graph.
Given a graph~$G$, we denote by $V(G)$ the set vertices of $G$, by $E(G)$
the set of edges of $G$, and by $\gamma(G)$ the number of connected components of $G$.
Let $G = (V, E)$ be a graph, given a subset of vertices $X \subseteq V(G)$ and a subset of
edges $E' \subseteq E$,
we define the following operation on $G$:
$G[X] = (X, \{uv \in E \mid u, v \in X\})$, $G - X = G[V \setminus X]$ and $G - E' = (V, E \setminus E')$.
We call a \emph{walk} in a graph $G$ an alternating sequence 
$W=v_0,e_1,v_1,\dots, e_r,v_r$ of vertices and edges, where $v_0,\ldots,v_r\in V(G)$, $e_1,\ldots,e_r\in E(G)$ and $e_i=v_{i-1}v_i$ for $i\in[r]$; note that $W$ may go through the same vertices and edges several times. A walk without repeated vertices and edges is called a \emph{path}.
We say that $v_0$ and $v_r$ are \emph{end-vertices} of $W$ and $W$ is a \emph{$(v_0,v_r)$-walk}. We use $V(W)\subseteq V(G)$ to denote the set of vertices of $G$ visited by $W$ and denote by $E(W)$ the set of edges of $G$ that are in $W$.
Given a walk $W=v_0,e_1,v_1,\dots, e_r,v_r$, for $0 \leq i \leq j \leq r$, we call the
sequence $W' = v_i, e_{i+1}, \dot, e_j, v_j$ a \emph{subwalk} of $W$.

\medskip
\noindent
\textbf{Temporal graphs.}
A \emph{temporal graph} $\mathcal{G}$ over a set of vertices $V$ is a sequence $\mathcal{G} = (G_1, G_2, \dots, G_L)$ of graphs such that for all $t \in [L], V(G_t) = V$. We call $L$ the \emph{lifetime} of $\mathcal{G}$ and for $t \in [L]$, we call $G_t = (V, E_t)$ the \emph{snapshot graph} of $\mathcal{G}$ at \emph{time step} $t$. We might refer to $G_t$ as $\mathcal{G}(t)$. We call $G = (V, E)$ with $E = \cup_{t \in [L]}$ the \emph{underlying graph} of $\mathcal{G}$.
We denote by $V(\mathcal{G})$ the vertex set of the underlying graph $G$ of $\mathcal{G}$
and by $\gamma(\mathcal{G}) = \max_{t\in L} \gamma(G_t)$ the maximum number of connected components
over all snapshot graphs of $\mathcal{G}$. We write $V$ and $\gamma$ if the graph or temporal graph
is clear from the context.
For a  temporal  graph $\GL$, we define the \emph{total number of edge appearances} as  $\m(\G) = \sum_{i=1}^{L} |E(G_i)|$; we use $\m$ to denote this value if $\G$ is clear for the context.

In the following, we will be interested in temporal walks where the agent has infinite speed within a snapshot graph. Those temporal walks are called \emph{non-strict temporal walks} in~\cite{DBLP:conf/sand/ErlebachS22}.
In~\cite{DBLP:conf/sand/ErlebachS22}, a non-strict temporal walk is define as a sequence of connected components. However, it is more convenient for us to consider a \emph{non-strict temporal walk} as a \emph{monotone} walk in the underlying graph. For a walk $W=v_0,e_1,v_1,\dots, e_r,v_r$ in the underlying graph $G$ of $\GL$, we say that $W$ is \emph{monotone} if there are $t_1,\ldots,t_r\in[L]$ with  $1\leq t_1\leq\cdots\leq t_r\leq L$ such that $e_{t_i}\in E(G_{t_i})$ for each $i\in[r]$.
The definition of a non-strict temporal walk immediately implies the following observation.

\begin{observation}\label{obs:equiv}
Given a temporal graph $\G$ and a vertex $x$,  $\G$ has a non-strict temporal walk starting in $x$ that visits the vertices of a set $X$ if and only if the underlying graph $G$ has a monotone $(x,y)$-walk $W$ for some $y\in V(G)$ such that $X=V(W)$.  
\end{observation} 
For the remainder of this paper, we are mainly interested in the computational problem of finding monotone walks that visit all vertices of a temporal graph. We might also call such a walk an \emph{exploration schedule} or simply an \emph{exploration}.
\begin{definition}[\textsc{Non-Strict Temporal Exploration (NS-TEXP)}]
	\ \\
	Input: Temporal graph $\mathcal{G} = (G_1, G_2, \dots, G_L)$, vertex $v \in V(\mathcal{G})$.\\	
    Question: Does there exist a monotone walk  in $\mathcal{G}$ that starts in $v$ and visits all vertices in $V(\mathcal{G})$?
\end{definition}

We further consider a more general variant of the \textsc{NS-TEXP} problem were we ask for a monotone walk that visits at least $k$ vertices (instead of $|V|$). We refer to this problem as \textsc{$k$-arb NS-TEXP}.
\begin{definition}[\textsc{$k$-arbitrary Non-Strict Temporal Exploration ($k$-arb NS-TEXP)}]
	Input: Temporal graph $\mathcal{G} = (G_1, G_2, \dots, G_L)$, vertex $v \in V(\mathcal{G})$, integer $k$.\\	
	Question: Does there exist a monotone walk in $\mathcal{G}$ that starts in $v$ and visits at least $k$ vertices?
\end{definition}

We will further generalize the problem by considering the weighted version of \textsc{$k$-arb NS-TEXP}. We refer to this problem as \textsc{Weighted $k$-arb NS-TEXP}.
\begin{definition}[\textsc{Weighted $k$-arbitrary Non-Strict Temporal Exploration \linebreak[5](Weighted $k$-arb NS-TEXP)}]\ \\
    Input: Temporal graph $\mathcal{G} = (G_1, G_2, \dots, G_L)$, positive-valued weight function $w\colon V(\mathcal{G}) \to \mathbb{N}_{>0}$, vertex $v \in V(\mathcal{G})$, integer $k$.\\	
	Question: Does there exist a monotone walk $W$ in $\mathcal{G}$ such that $W$ starts in $v$ and 
    $w(V(W)) = \sum_{u \in V(W)} w(u) \geq k$?
\end{definition}
Note that \textsc{$k$-arb NS-TEXP} is a special case of \WEX{} with $w(x)=1$ for every $x\in V$.

\medskip
\noindent
\textbf{Parameterized complexity and kernelization.}
We refer to the book~\cite{CyganFKLMPPS15} for a detailed introduction to the field, see also the recent book on kernelization~\cite{kernelizationbook19}. Here, we only briefly introduce basic notions.

\begin{toappendix}
A \emph{parameterized problem} is a language $Q\subseteq \Sigma^*\times\mathbb{N}$ where $\Sigma^*$ is the set of strings over a finite alphabet $\Sigma$. Respectively, an input  of $Q$ is a pair $(I,k)$ where $I\subseteq \Sigma^*$ and $k\in\mathbb{N}$; $k$ is the \emph{parameter} of  the problem. 

A parameterized problem $Q$ is \emph{fixed-parameter tractable} (\classFPT) if it can be decided whether $(I,k)\in Q$ in  $f(k)\cdot|I|^{\Oh(1)}$ time for some function $f$ that depends on the parameter $k$ only. Respectively, the parameterized complexity class \classFPT is composed by  fixed-parameter tractable problems.

Parameterized complexity theory also provides tools to rule-out the existence of \classFPT algorithms under plausible complexity-theoretic assumptions. For this,  a hierarchy of parameterized complexity classes $$\classFPT\subseteq \classW{1}\subseteq  \classW{2}\subseteq  \cdots\subseteq \classXP$$ was introduced in ~\cite{DowneyFellows1992b}, and it was conjectured that the inclusions are proper. The basic way to show that it is unlikely that a parameterized problem admit an \classFPT algorithm is to show that it is $\classW{1}$ or $\classW{2}$-hard.

 A \emph{data reduction rule}, or simply, reduction rule, for a parameterized problem~$Q$ is  a function~$\phi \colon \Sigma^{\ast} \times \mathbb{N} \rightarrow \Sigma^{\ast} \times \mathbb{N}$ that maps an instance~$(I, k)$ of~$Q$ to an 
equivalent instance~$(I', k')$ of~$Q$ such that $\phi$ is computable in time polynomial in~$|I|$ and~$k$. We say that two instances of~$Q$ are \emph{equivalent} if the following holds: $(I, k) \in Q$ if and only if ~$(I',k') \in Q$. We refer to  this property of the reduction rule $\phi$, that it translates an instance to an equivalent one,  as to the \emph{safeness}  of the reduction rule.

Informally, \emph{kernelization} is a { {preprocessing algorithm}} that consecutively applies various data reduction rules in order to shrink the instance size as much as possible. A preprocessing algorithm takes as input an instance $(I,k)\in\Sigma^{*}\times\mathbb{N}$ of $Q$, works in polynomial in~$|I|$ and~$k$ time, and returns an equivalent instance $(I',k')$ of $Q$. The quality of a preprocessing algorithm  $\mathcal{A}$ is measured by the size of the output. More precisely,  the {\em{output size}} of a preprocessing algorithm $\mathcal{A}$ is a function $\textrm{size}_{\mathcal{A}}\colon \mathbb{N}\to \mathbb{N}\cup \{\infty\}$ defined as follows:
$$\textrm{size}_{\mathcal{A}}(k) = \sup \{ |I'|+k'\ \colon\ (I',k')=\mathcal{A}(I,k),\ I\in \Sigma^{*}\}.$$
A {\em{kernelization algorithm}}, or simply a {\em{kernel}}, for a parameterized problem $Q$ is a preprocessing algorithm $\mathcal{A}$ that, given an instance $(I,k)$ of $Q$, works in polynomial in~$|I|$ and~$k$ time and returns an equivalent instance $(I',k')$ of $Q$ such that   $\textrm{size}_{\mathcal{A}}(k)\leq g(k)$ for some computable function $g\colon \mathbb{N}\to \mathbb{N}$. It is said that $g(\cdot)$ is the \emph{size} of a kernel.
If  $g(\cdot)$ is a polynomial function, then we say that $Q$ admits a {\emph{polynomial  kernel}}.

\medskip
\noindent
{\bf Cross-composition.}  It is well-known that every \classFPT problem admits a kernel \cite{DowneyF99}. However,  up to some reasonable complexity assumptions, there are \classFPT problems that have no polynomial kernels. 
In particular, we are using the cross-composition technique introduced in~\cite{BodlaenderDFH09} to show that a parameterized problem does not admit a polynomial kernel unless $\classNP\subseteq\classCoNP/{\rm poly}$.

To define the cross-composition, we start with the definition of polynomial equivalence relation.

\begin{definition}[Polynomial equivalence relation] 
An equivalence relation $\cal R$ on the set $\Sigma^*$ is called a {\em{polynomial equivalence relation}} if the following conditions are satisfied:
\begin{enumerate}[(a)]
\item There exists an algorithm that, given strings $x,y\in \Sigma^*$, resolves whether $x$ is equivalent to $y$ in time polynomial in $|x|+|y|$.
\item
For any finite set $S \subseteq \varSigma^*$ the equivalence relation $\cal R$ partitions the elements of $S$ into at most $(\max_{x \in S} |x|)^{\Oh(1)}$ classes.
\end{enumerate}
\end{definition}

The cross-decomposition of a language into a parameterized language is defined as follows. 

\begin{definition}[Cross-composition] \label{def:cross_comp}
Let $L\subseteq \Sigma^*$ be a language and $Q\subseteq \Sigma^*\times \mathbb{N}$ be a parameterized language. We say that $L$ {\em{cross-composes}} into $Q$ if there exists a polynomial equivalence relation $\cal R$ and an algorithm $\cal A$, called a {\em{cross-composition}}, satisfying the following conditions. The algorithm $\cal A$ takes as input a sequence of strings $x_1,x_2,\ldots,x_t\in \Sigma^*$ that are equivalent with respect to $\cal R$, runs in time polynomial in $\sum_{i=1}^t |x_i|$, and outputs one instance $(y,k)\in \Sigma^*\times \mathbb{N}$ such that:
\begin{enumerate}[(a)]
\item $k\leq p(\max_{i=1}^t |x_i|+\log t)$ for some polynomial $p(\cdot)$, and
\item $(y,k)\in Q$ if and only if there exists at least one index $i$ such that $x_i\in L$.
\end{enumerate}
\end{definition}

The machinery for ruling out the existence of polynomial kernels is based on the following theorem of Bodlaender et al. \cite{BodlaenderDFH09}.  
It says that if some parameterized problem admits a cross-composition into it of some   {\NP}-hard language, then the existence of a polynomial kernel for such a problem is very unlikely.

\begin{proposition}[\cite{BodlaenderDFH09}]\label{crosscollapse}
  Let $L \subseteq \varSigma^*$ be an  {\NP}-hard language. 
 If $L$ cross-composes into a parameterized problem $Q$ and $Q$ has a polynomial kernel,  
  then  $\classNP\subseteq\classCoNP/{\rm poly}$. 
\end{proposition}
\end{toappendix}

\section{Lower Bounds}
\subsection*{Lower Bounds for \textsc{NS-TEXP} - \NP-hardness of Restricted Cases}
It was stated as an open problem in~\cite{DBLP:conf/sand/ErlebachS22}, whether \textsc{NS-TEXP} is in \FPT{} with the parameter maximum number of components in any snapshot graph, i.e., with parameter $\gamma$.
We answer this question negatively by showing that \textsc{NS-TEXP} is \NP-complete for $\gamma\geq 5$.

\begin{theorem}
	\label{thm:const-c}
	\textsc{NS-TEXP} is \NP-complete for $\gamma\geq 5$.
\end{theorem}
%
\begin{proof}
	We give a reduction from the satisfiability problem \textsc{SAT} which asks if a given Boolean formula has a satisfying variable assignment. Let $\varphi = \{c_1, c_2, \dots, c_m\}$ be a Boolean formula in conjunctive normal form over the variable set $X = \{x_1, x_2, \dots, x_n\}$. We construct from $\varphi$ a temporal graph $\mathcal{G}$, where each snapshot graph has 4 or 5 connected components, such that $\mathcal{G}$ has a monotone walk 
	 that visits all vertices in $V(\mathcal{G})$ if and only if $\varphi$ is satisfiable.

	The main idea of the construction is the following.
	In $V(\mathcal{G})$, we have a vertex for each clause, two vertices for each variable $x$, corresponding to $x$ and $\neg x$, a control vertex $\widehat{x}$ for each variable~$x$, and one single control vertex $\widehat{c}$ for the set of clauses. 
	The sequence of snapshot graphs alternates between having four and five connected components.
	In the case of four connected components, one component collects all clause vertices, one component collects all variable vertices $x$ and $\neg x$, one component collects all not yet processed control vertices $\widehat{x}$ and $\widehat{c}$, and one component collects all processed control vertices $\widehat{x}$. 
	In the case of five connected components, an additional component collects a negated literal $\neg x$ of a variable $x$ together with all clauses containing $\neg x$. In this step, the clauses containing $x$ are incorporated into the variable component, which still contains $x$. This will allow us to choose a variable assignment with the exploration schedule. 
	In the next time step, and only in this time step, the control vertex $\hat{x}$ is contained in the variable component. For all later time steps, $\hat{x}$ is contained in the component collecting the processed control vertices.
	Thereby, the control vertices ensure that we return to the component containing the variables in each snapshot consisting of four connected components. 
\begin{toappendix}

	We now give a more formal construction.
    For this construction, we are interesting in the connected components of each snapshot graph
    but not on the actual structure of the connected components. For simplicity, we
    define the connected components as set of vertices and assume that each connected
    component forms a clique. For a subset $S \subseteq V$, we denote
    by $K(S)$ the set of all possible edges between vertices in $S$. Thereby, $(S, K(S))$
    is a clique.
	Let $\mathcal{G} = (G_1, G_2, \dots, G_L)$ with $L = 2n +1$ be the temporal graph constructed from $\varphi = \{c_1, c_2, \dots, c_m\}$. We define $V(\mathcal{G}) = \{c_1, c_2, \dots, c_m\} \cup \{x_i, \neg x_i, \widehat{x_i}\mid x_i \in X\} \cup \{\widehat{c}\}$.
    For $E(G_1) = K(C_{1,1}) \cup K(C_{1,2}) \cup K(C_{1,3})$ we set $C_{1,1} = \{\widehat{c}\} \cup \{\widehat{x_i} \mid x_i \in X\}$, $C_{1,2} = \{c_1, c_2, \dots, c_m\}$, and $C_{1,3} = \{x_i, \neg x_i \mid x_i \in X\}$.
	The start vertex is set to $x_1$.

	For the next time step, we define the edges of the snapshot graph as 
	$E(G_2) = K(C_{2,1}) \cup K(C_{2,2}) \cup K(C_{2,3}) \cup K(C_{2,4})$ where
    $C_{2,1} = C_{1,1} = \{\widehat{c}\} \cup \{\widehat{x_i} \mid x_i \in X\}$,
    $C_{2,2} = \{c_1, c_2, \dots, c_m\} \setminus \{c \in \varphi \mid x_1 \in c \vee \neg x_1 \in c\}$,
    $C_{2,3} = (\{x_i, \neg x_i \mid x_i \in X\} \setminus \{\neg x_1\}) \cup \{c \in \varphi \mid x_1 \in c\}$ and
    $C_{2,4} = \{\neg x_1\} \cup \{c \in \varphi \mid \neg x_1 \in c\}$.
    The intuition is that the exploration schedule visits the vertices in $C_{2,3}$ after
    visiting the vertices in $C_{1,3}$ if the variable $x_1$ gets assigned with \texttt{true}
    and otherwise, visits $C_{2,4}$ after $C_{1,3}$ if $x_1$ gets assigned with \texttt{false}.
    Note that no other connected component is reachable from $C_{1,3}$ as it has an empty intersection with $C_{1,3}$.
	
	In the next time step, we force the exploration schedule to return to the third connected component
    by passing the control vertex $\widehat{x_1}$ through this connected component. Therefore, we define
	$E(G_3) = K(C_{3,1}) \cup K(C_{3,2}) \cup K(C_{3,3})$ with $C_{3,1} = (\{\widehat{c}\} \cup \{\widehat{x_i} \mid x_i \in X\}) \setminus \{\widehat{x_1}\}$, $C_{3,2} = \{c_1, c_2, \dots, c_m\}$, and $C_{3,3} = \{x_i, \neg x_i \mid x_i \in X\} \cup \{\widehat{x_1}\}$. 
	
	For the remaining time steps, we alternate between the structure of the second and third time step with an additional connected component that collects the already processed control vertices $\widehat{x_i}$. As this connected component is monotone growing, it acts as a sinc for the temporal walk, i.e., we could not leave this connected component if we ever enter it. 
	Additionally, the first connected component containing the not yet processed control vertices is monotone shrinking, making it non-accessible from the start vertex of the temporal walk. 
	The idea is now that the last control vertex $\widehat{c}$ will only leave the first component into the variable component in the last time step enforcing us to not go into the sinc component of processed control vertices and thereby enforcing the exploration schedule to return to the variable component for each odd time step $t \geq 3$. We formalize this by defining the snapshot graphs $G_j$ for $3 < j < 2n+1$ as follows. 
	
	First consider the case that $j$ is even. We define
    $E(G_j) = K(C_{j,1}) \cup K(C_{j,2}) \cup K(C_{j,3}) \cup K(C_{j,4}) \cup K(C_{j,5})$ with
    $C_{j,1} = \{\widehat{c}\} \cup \{\widehat{x_i} \mid x_i \in X, i \geq \frac{j}{2}\}$,
    $C_{j,2} = \{c_1, c_2, \dots, c_m\} \setminus \{c \in \varphi \mid x_{\frac{j}{2}} \in c \vee \neg x_{\frac{j}{2}} \in c\}$,
    $C_{j,3} = (\{x_i, \neg x_i \mid x_i \in X\} \setminus \{\neg x_{\frac{j}{2}}\}) \cup \{c \in \varphi \mid x_{\frac{j}{2}} \in c\}$,
    $C_{j,4} = \{\neg x_{\frac{j}{2}}\} \cup \{c \in \varphi \mid \neg x_{\frac{j}{2}} \in c\}$, and
    $C_{j,5} = \{\widehat{x_i} \mid x_i \in X, i < \frac{j}{2}\}$.
	
	For the case that $j$ is odd, we define
	$E(G_j) = K(C_{j,1}) \cup K(C_{j,2}) \cup K(C_{j,3}) \cup K(C_{j,4}) \cup K(C_{j,5})$ with 
	$C_{j,1} = \{\widehat{c}\} \cup \{\widehat{x_i} \mid x_i \in X, i > \frac{j}{2}\}$, $C_{j,2} = \{c_1, c_2, \dots, c_m\}$, $C_{j,3} = \{x_i, \neg x_i \mid x_i \in X\} \cup \{\widehat{x_{\frac{j-1}{2}}}\}$, $C_{j,4} = \emptyset$, and $C_{j,5} = \{\widehat{x_i} \mid x_i \in X, i < \frac{j-1}{2}\}$.
	
	Finally, for the last time step $L = 2n+1$, we define 
	$E(G_L) = K(C_{L,1}) \cup K(C_{L,2})\cup K(C_{L,3}) \cup K(C_{L,4}) \cup K(C_{L,5})$ with
	$C_{L,1} = \emptyset$, $C_{L,2} = \{c_1, c_2, \dots, c_m\}$, $C_{L,3} = \{x_i, \neg x_i \mid x_i \in X\} \cup \{\widehat{c}\}$, $C_{j,4} = \emptyset$, and $C_{j,5} = \{\widehat{x_i} \mid x_i \in X\}$.
	
	Let us now analyze how a potential exploration schedule for $\mathcal{G}$ can look like. Recall that for each time step $t$, the first connected component $C_{t,1}$ is monotonously shrinking, i.e., for $t, t' \in [L]$ with $t \leq t'$ it holds that $C_{t, 1} \supseteq C_{t', 1}$. As the start vertex of the exploration schedule is contained in the third connected component, this means that (i) any exploration schedule cannot visit a first connected component as the connected component $C_{t,1}$ never has a non empty intersection with any connected component $C_{t-1, \ell}$ for $1 < t \leq L$, and $\ell \neq 1$.
	
	Further, recall that the fifth connected component $C_{t, 5}$ is monotonously growing, i.e., for $t, t' \in [L]$ with $t \leq t'$ it holds that $C_{t, 1} \subseteq C_{t', 1}$. Therefore, (ii) an exploration schedule cannot leave any fifth connected component. 
	
	Now consider the vertex $\widehat{c}$. For all time steps $t \in [L-1]$, this vertex is contained in the first connected component $C_{t,1}$ and therefore not reachable from the start vertex. The only time step in which $\widehat{c}$ is in another connected component is the last time step $t=2n+1$ in which $\widehat{c}$ is contained in the third connected component $C_{t,3}$ together with the variable vertices. As an exploration schedule need to reach $\widehat{c}$ and $\widehat{c}$ is never contained in a fifth connected component, observation (ii) implies, that we must never enter a fifth connected component. As all elements in $\{\widehat{x_i} \mid x_i \in X\} \cup \{\widehat{x}\}$ do only appear in a first, third or fifth connected component, observation (i) implies that (iii) we need to visit them in a third connected component. 
	
	We further need to visit all elements in $\{c_1, c_2, \dots, c_m\}$. Assume, we visit some of them for the first time in some second connected component in a time step $t$. By the construction of $\mathcal{G}$, the second connected component is only reachable from a third or fourth connected component in an odd time step. Note that the odd time step $t$ is the only time step in which the control vertex $\widehat{x_{\frac{t-1}{2}}}$ is contained in a third connected component. As by assumption the exploration schedule visits the second connected component in time step $t$ it will not visit $\widehat{x_{\frac{t-1}{2}}}$ in time step $t$. But observation (iii) implies that then, $\widehat{x_{\frac{t-1}{2}}}$ cannot be reached during the remaining walk. Hence, any exploration schedule cannot visit any second connected component and the vertices $\{c_1, c_2, \dots, c_m\}$ must be visited in a third or fourth connected component in an even time step.
	
	We already observed that any exploration schedule must be in the third component in any odd time step. From this component, it is only possible to visit some vertex $c_j$ by traversing to a component containing a literal that is contained in $c_j$. In each even time step~$t$, it is possible to visit those clause vertices $c_j$ that contain a literal from the variable $x_{\frac{t}{2}}$. Thereby, we can either visit those containing the positive literal, or those containing the negative literal. An exploration schedule visiting all clause vertices $\{c_1, c_2, \dots, c_m\}$ thereby corresponds to a satisfying variable assignment for $\varphi$. Conversely, a satisfying variable assignment for $\varphi$ gives us an exploration schedule for $\mathcal{G}$ by traversing to the component containing the satisfied literal of the variable $x_i$ in time step $x_{2i}$.	
\end{toappendix}
\end{proof}

We now shift our focus to restricting the graph class of the underlying graph. 
We show that \textsc{NS-TEXP} is \NP-hard restricted to temporal graphs where the underlying graph is a tree consisting of two stars connected with an edge, or in other words, trees of depth two. 
The vertex cover number of such  trees is two. 
As a corollary, we obtain with a simple adaption of the construction that \textsc{NS-TEXP} is \NP-hard even if every edge appears at most once, or is non-present in at most one time step and the underlying graph is a tree.
\begin{theorem}
	\label{thm:treeNP}
	\textsc{NS-TEXP} is \NP-complete for temporal graphs where the underlying graph consists of two stars connected with a bridge or respectively, trees of depth two.
\end{theorem}
\begin{toappendix}
\begin{proof}
    We give a reduction from \textsc{Monotone 3SAT}~\cite{DBLP:journals/iandc/Gold78,DBLP:journals/dam/Li97a}, which asks if a given Boolean formula in conjunctive normal form with three literals per clause (3CNF) where each clause either contains only positive or only negative literals (monotone) has a satisfying variable assignment. Let $\varphi = \{c_1, c_2, \dots, c_m\}$ be a monotone Boolean formula in 3CNF over the variable set $X = \{x_1, x_2, \dots, x_n\}$. We construct from $\varphi$ a temporal graph $\mathcal{G}$ with an underlying graph that meets the claimed restriction, such that $\mathcal{G}$ has a non-strict temporal walk starting in some vertex $v$ that visits all vertices in $V(\mathcal{G})$ if and only if $\varphi$ is satisfiable.
	
	Let $V(\mathcal{G}) = \{c_1, c_2, \dots, c_m\} \cup \{\top, \bot\}$. The vertices $\top$ and $\bot$ will be the center of the two stars. We connect $\top$ in the underlying graph with each vertex $c_i$ that corresponds to a positive clause, i.e., a clause containing only positive literals. Respectively, we connect $\bot$ in the underlying graph with each vertex $c_j$ corresponding to a negative clause. Finally, we connect $\top$ and $\bot$ with an edge and set $\top$ as the start vertex $v$. 
	
	The variables of $\varphi$ are not represented as vertices but instead encoded in the time steps as follow. For each odd time step in $[2n]$ we only let the edge between $\top$ and $\bot$ be present. For each even time step $t \in [2n]$, we let exactly those edges be present that connect $\top$ with a clause $c_i$ containing the positive literal $x_{t/2}$; and that connect $\bot$ with a clause $c_j$ containing the negative literal $\neg x_{t/2}$.
	Observe that as $\varphi$ is monotone, each clause vertex has degree one and as $\varphi$ is in 3CNF, each edge incident to a clause vertex appears at most three times. Further, the constructed underlying graph consists of two stars with a connecting bridge.
	
	Intuitively, in each odd time step $t$, we can traverse from $\top$ to $\bot$ or vise versa and thereby choose the assignment of the variable $x_{(t+1)/2}$. In the next even time step, we can then visit all clauses containing a literal that is satisfied under the chosen assignment of $x_{(t+1)/2}$. Clearly, if $\varphi$ is satisfiable, the described strategy gives a non-strict temporal walk that visits all vertices. Conversely, let $W$ be a non-strict temporal walk in $\mathcal{G}$ that visits all vertices in $V(\mathcal{G})$. We obtain from $W$ a variable assignment for $\varphi$ as follows. For each even time step $t \in [2n]$, we assign $x_{t/2}$ with true if $W$ is in a component containing $\top$ in time step $t$, we assign $x_{t/2}$ with false if $W$ is in a component containing $\bot$ in time step $t$, and arbitrarily with true or false if $W$ is in any other (singleton) component. Intuitively, the last case corresponds to a situation where the assignment of $x_{t/2}$ is irrelevant for the satisfiabiltiy of $\varphi$. As by assumption, $W$ visits all vertices, it visits all clause vertices. As we read the variable assignment from $W$ according to the visits of the clause vertices, it is clear that the obtained variable assignment satisfies $\varphi$.	
\end{proof}
\end{toappendix}


\begin{corollary}\label{cor:single}
	\textsc{NS-TEXP} is \NP-complete even if every edge of the input temporal graph appears at most once.
\end{corollary}
\begin{toappendix}
\begin{proof}
	In the construction of \autoref{thm:treeNP} the edges between clause vertices and $\{\top, \bot\}$  appear at most three times. The only edge that appears a linear number of times is the edge $\{\top, \bot\}$. We slightly change the construction by deleting the edge $\{\top, \bot\}$ and introducing $n$ many new vertices $v_1, v_3, \dots, v_{2n-1}$ that are connected with an edge to $\top$ and with an edge to $\bot$. Then, in each odd time step $t\in [2n]$, we only let the edges incident to $v_{t}$ be present. Thereby, those edges appear only once. 
	It can easily be observed, that if there is an exploration schedule for the constructed graph, then there is one that never waits in a clause vertex. 
	Hence, we visit the new vertices connected to $\top$ and $\bot$ in the odd time steps.
	
	It remains to replace the edges connecting $\top$ and $\bot$ with the clause vertices. Those edges appear at most three times. Again, we replace each of those edges by a path of length two for each time step in which the original edge appears. In contrast to the edge between $\bot$ and $\top$, we are not guaranteed to be in the connected component containing an edge incident to a clause vertex in the time step in which it appears. Thereby, we are not guaranteed to visit all of the newly introduced vertices adjacent to clause vertices. 
	To counter this, we extend the sequence of snapshot graph by one time step at the end. We connect the newly introduced vertices which are adjacent to clause vertices in a cycle which edges are only present in this one last time step. Thereby, we ensure that we visit all introduced vertices at the end without being able to visit additional clause vertices.
	
	Note that this modification yields an underlying graph that is no longer a tree.
\end{proof}
\end{toappendix}


\begin{corollary}
	\textsc{NS-TEXP} is \NP-complete even if the underlying graph is a tree and every edge of the input temporal graph is not present in at most one time step.
\end{corollary}
\begin{toappendix}
\begin{proof}
	We start from the construction in the proof of \autoref{thm:treeNP}. Recall that the lifetime of the constructed temporal graph is $2n$. Now we replace each edge by a path of length $4n$ (on new vertices).
	We define the appearance of the edges in the paths as follows. Let $\pi$ be a path replacing the edge $e$. Then, for each time step $t$ in which $e$ is present, we set all edges in $\pi$ to be present. If $e$ is not present in time step $t$, then we set the $t$'th and $4n-t$'th edge in $\pi$ to be not present and all other edges in $\pi$ to be present. Thereby, the path $\pi$ cannot be crossed in time step $t$. As the non existing edges are wandering towards the middle of the path with increasing time steps, we avoid glitching through the path by waiting on one side of a currently non-existing edge.
	
	Finally, we need to ensure that all newly introduced vertices can be visited while keeping the properties of the reduction. It is clear that by the introduced modification we cannot visit any additional clause vertex. Further, the vertices of the path between $\bot$ and $\top$ are visited in the first time step. For the vertices on a path to a clause vertex, observe that by the construction, we are forced to visit all clause vertices eventually, and we can visit them in the same time steps as in the original construction. As the clause vertices are of degree one, for each clause vertex, there is only one path leading to it. Hence, all vertices on this path will be visited when we visit the respective clause vertex. This concludes the proof.
\end{proof}

\subsubsection*{Lower Bounds for \textsc{NS-TEXP} - Kernels}
\end{toappendix}

For the combined parameter $\gamma + L$, there exists a trivial \FPT-algorithm for the problem \textsc{$k$-arb NS-TEXP} (and so for \textsc{NS-TEXP}) as this parameter bounds the size of a search tree, where $\gamma$ is the branching factor and $L$ is the depth of the tree. In this tree we consider for each time step all connected components that are reachable from the current connected component, starting with the component containing the start vertex $v$.
\begin{observation}\label{obs:fpt_k_arb}
    \textsc{$k$-arb NS-TEXP} parameterized by $\gamma + L$ can be solved in
    $\Oh(\gamma^L n^{\Oh(1)})$ time.
\end{observation}
In contrast, using a reduction from \textsc{Hitting Set}, we obtain that \textsc{NS-TEXP} does not have a polynomial kernel in the same parameters unless \NP{} $\subseteq$ \coNP/\poly. In fact, we show a stronger statement: if $\gamma$ is constant, there does not exists a polynomial kernel in $L$ under the same assumption.
\begin{theorem}%
	\label{thm:no-kernel-L}
	Unless \NP $\subseteq$ \coNP/\poly, there does not exist a polynomial kernel for \textsc{NS-TEXP} in the parameter $L$ for $\gamma \geq 6$ being constant.
\end{theorem}

\begin{toappendix}
We show the claim by a reduction from the \textsc{Hitting Set} problem.
\begin{definition}[\textsc{Hitting Set}]
		\ \\
	Input: Universe $\mathcal{U} = \{x_1, x_2, \dots, x_n\}$, collection of sets $\mathcal{S} = \{S_1, S_2, \dots, S_m\} \subseteq \mathcal{U}$, integer $k$.\\	
	Question: Does there exist $H \subseteq \mathcal{U}$ such that $|H|\leq k$ and $H \cap S_i \neq \emptyset$ for all $1\leq i \leq m$?
\end{definition}
It was shown in~\cite[Thm. 5.3]{DBLP:journals/talg/DomLS14} that \textsc{Hitting Set} parameterized by the size of the universe has no polynomial kernel unless \NP $\subseteq$ \coNP /\poly.

\begin{proof}
Let $\mathcal{U} = \{x_1, x_2, \dots, x_n\}$, $\mathcal{S} = \{S_1, S_2, \dots, S_m\} \subseteq \mathcal{U}$, and $k$ be an instance of \textsc{Hitting Set}. We construct a temporal graph $\mathcal{G}$ with vertex set $V(\mathcal{G}) = \mathcal{U} \cup \mathcal{S} \cup \bigcup_{1 \leq i \leq k} \mathcal{U}^i$, where $\mathcal{U}^i$ is an annotated copy of $\mathcal{U}$ where each element $x_j$ is annotated as $x_j^i$.

The lifetime $L$ of $\mathcal{G}$ is divided into $k$ phases where each phase consists of a choice of one element $x \in \mathcal{U}$ of the hitting set and the visit of all vertices corresponding to sets in $\mathcal{S}$ containing the element $x$.
During each of the $k$ phase, a annotated copy $\mathcal{U}^i$ is used as control vertices to enforce some structure of the exploration schedule.
Similar to \autoref{thm:const-c}, we are interesting in the vertices of the connected components and not in theirs structures. Thereby, we define the connected components as set of vertices and assume that each connected component forms a clique. As in \autoref{thm:const-c}, we keep a constant number of components throughout the construction.

Now, we give an intuitive description of the vertices contained in each connected component in the snapshot graph at time step $t$. The first connected component $C_{t,1}$ contains the unused control vertices of $\bigcup_{1 \leq i \leq k}\mathcal{U}^i$ at time step $t$. This component will only shrink through the lifetime of $\mathcal{G}$.
The second connected component $C_{t,2}$ contains all elements from $\mathcal{S}$ which are currently not processed.
The last connected component $C_{t,6}$ contains the elements from $\bigcup_{1 \leq i \leq k}\mathcal{U}^i$ that have already been used and will only increase through the lifetime of $\mathcal{G}$.
The three other connected components are used to select an element $x$ of $\mathcal{U}$ and process the set of $S$ containing $x$ during each of the $k$ phases. They are used as follows.
At the start of a phase, the third connected component $C_{t,3}$ contains all the elements
from $\mathcal{U}$. Then, during a phase, at each time step, an element $x \in \mathcal{U}$
is moved from the third component $C_{t,3}$ to the forth $C_{t,4}$ together with all the set
in $S$ containing $x$. In the next time step, $x$ is move to the fifth connected component
$C_{t+1,5}$ where it will stay until the end of the phase.
We now define the components $C_{t,3}$, $C_{t,4}$, and $C_{t,5}$ more formally.

Each phase contains $(n+1)$ time steps.
Let $t = (i-1) \cdot (n+1)$, for $i \in [k]$ be the first time step of the $i$th phase.
By staying on an element $x$ while it leaves the third component within next $n$ time step, the $i$th choice of an element of $\mathcal{U}$ is taken. Therefore, in time step $t$, $C_{t,3} = \mathcal{U} \cup \{x_n^{i-1} \mid \text{if } i-1>0\}$, i.e., all elements of $\mathcal{U}$ are available and no element is currently processed.
The role of the added element $x_n^{i-1}$ becomes clear later. It can be thought of as an overhang from the previous phase.
The components $C_{t,4}$ and $C_{t,5}$ are empty.
For the time steps $t+j$ with $j \in [n]$, we define 
$C_{t+j,3} = \{x_{\ell} \in \mathcal{U} \mid \ell > j\} \cup \{x_{j-1}^i \mid \text{if }j-1 >0\}$,
$C_{t+j,4} = \{x_j\} \cup \{S \in \mathcal{S} \mid x_j \in S \}$, and
$C_{t+j,5} = \{x_{\ell} \in \mathcal{U} \mid \ell < j\} \cup \{x_{j}^i \mid \text{if }j < n\}$.

Finally, we define for the last time step $t = k \cdot (n+1)$, $C_{t,1} = \emptyset$, $C_{t,2} = \mathcal{S}$, $C_{t,3} = \mathcal{U} \cup \{x_n^k\}$, $C_{t,4} = \emptyset$, $C_{t,5} = \emptyset$, and $C_{t,6} = {\bigcup_{1 \leq i \leq k} \mathcal{U}^i}\setminus \{x_n^k\}$.
By setting any $x \in C_{1,3}$ as the start vertex, we conclude the construction of the
reduction.
Note that $L = k\cdot (n+1) +1$ which is of order $\Oh(n^2)$.

It is clear that if there exists a hitting set $H$ of size $k$, then by fixing an order on the elements of $H$ and following the $i$th element during the $i$th phase of $\mathcal{G}$ describes a non-strict temporal walk through $\mathcal{G}$ that visits all elements in $\mathcal{S}$, as $H$ is a hitting set, and all elements in $\mathcal{U}$, as we start in the component $C_{1,3}$.
By following an element $x$ through the graph, we mean that we always move to the component containing $x$. We further see all elements in $\mathcal{U}^i$ for $1 \leq i \leq k$ due to the following observation. Let $x_i$ be the $i$th element in $H$. Then, in the $i$th phase we follow the vertex $x_i$ through the temporal graph. By this, we see the vertices $x_{\ell}^i$ for $\ell < i$ in the time steps $t+\ell$ in the components $C_{t+\ell,3}$ and the vertices $x_{\ell'}$ for $i \leq \ell'$ in the time steps $t+\ell'+1$ in the components $C_{t+\ell'+1,5}$ except for $x_n^{i}$ which is seen in the third component at the start of the phase $i+1$.

For the other direction, assume $W$ is an exploration schedule for~$\mathcal{G}$. Observe that as $W$ starts in the component $C_{1,3}$, it can never reach a component $C_{t,1}$ for any~$t$, as the first component is only shrinking. Further, if $W$ visits a component $C_{t,6}$, then, for all $t' > t$, it can only visit the components $C_{t', 6}$, as the sixth component is only growing.
As the element $x_n^k$ is processed only in the very last step, it never reaches the sixth component. Hence, we cannot reach $x_n^k$ on any temporal walk that visits a sixth component.
Therefore, we are restricted to visit all vertices through components $C_{t,2} - C_{t,5}$. We show next that any exploration schedule cannot reach a second component $C_{t,2}$. Note that as we can neither visit a component $C_{t,1}$, nor a component $C_{t,6}$, we must `catch' the elements from $\bigcup_{1 \leq i \leq k}\mathcal{U}^i$ as they pass through the components $C_{t,3}$, or $C_{t,5}$. If we miss an element $x_\ell^i$ in a component $C_{t,3}$, then the only option to still visit it is in the component $C_{t+1,5}$. Hence, if we visit any component $C_{t,4}$ (and thereby, miss $C_{t,3}$), we are forced to visit the component $C_{t+1,5}$ in the next step, otherwise we can no longer visit all vertices. As we can only visit any second component $C_{t,2}$ over the component $C_{t-1,4}$ this implies that we cannot reach any second component in an exploration schedule and further cannot visit two fourth components $C_{t,4}, C_{t+1,4}$ in a row without missing an element from $\bigcup_{1 \leq i \leq k}\mathcal{U}^i$. As within one phase, the fifth component is only growing, the fact that we cannot visit two times in a row a fourth component implies that we cannot visit any two forth components within one phase. Hence, we can collect the elements from $\mathcal{U}$ that appear in $W$ in a forth components into a set $H$. By the previous arguments, it follows that $|H| \leq k$, and as we cannot visit any second component, it follows that $H$ has a non-empty intersection with each set $S$ in $\mathcal{S}$, in other words, $H$ is a hitting set for $\mathcal{S}$.
\end{proof}
\end{toappendix}

The size of an instance of \textsc{NS-TEXP} can be bounded by $n^2 \times L$ and there exists \FPT{} algorithms in the parameter $L$ as well as in the parameter $n$~\cite{ErlebachS22Journal}. We already showed that under standard complexity assumption, there does not exists a polynomial kernel for the parameter $L$. 
The next result is obtained by a cross-composition from \textsc{NS-TEXP} into itself
\begin{toappendix}
using \autoref{crosscollapse}.
\end{toappendix} 

\begin{theorem}
	\label{thm:no-kernel-n}
	Unless $\coNP \subseteq \NP/\poly$, \textsc{NS-TEXP} does not admit a polynomial kernel parameterized by $n$.
\end{theorem}
\begin{toappendix}
\begin{proof}
    We show the claim using \autoref{crosscollapse} via a cross-composition from \textsc{NS-TEXP} into itself. Therefore, we reduce
    $\ell$ instances $(\mathcal{G}_1,v^1), (\mathcal{G}_2, v^2), \dots (\mathcal{G}_\ell, v^\ell)$ of \textsc{NS-TEXP} into one instance
    $(\mathcal{G}, v)$ of \textsc{NS-TEXP} such that $(\mathcal{G}, v)$ is a yes-instance
    if and only if at least one of the instances $(\mathcal{G}_i, v^i)$ for
    $1 \leq i \leq \ell$ is a yes-instance.
    First, we show that we can transform the instances $(\mathcal{G}_i, v^i)$, $1 \leq i \leq \ell$, such
    that every instance has the same number of vertices. Let $n = \max_i |V(\mathcal{G}_i)|$.
    If, for some $i$, $|V(\mathcal{G}_i)| < n $, we add a set $S$ of $n - |V(\mathcal{G}_i)|$
    new vertices to~$\mathcal{G}_i$, and we precede the sequence of snapshot graphs of $\mathcal{G}_i$ with a new snapshot graph containing two connected component, one containing $\{v^i\} \cup S$ and the other one
    containing the rest of the vertices. For all other snapshot graphs of $\mathcal{G}_i$,
    we add $S$ as a connected component.
    We call $\mathcal{G}_i'$ the newly obtained instance. It is easy to see that
    $\mathcal{G}_i'$ is a yes-instance if and only if $\mathcal{G}_i$ is a yes-instance.
    Therefore, we assume in the following that all instances have the same number of vertices, thus, without
    loss of generality, we can rename the vertex sets $V(\mathcal{G}_i)$ and assume that all instances have the same vertex set and the same start vertex $v_1$.
    From now on, let $(\mathcal{G}_1,v_1), (\mathcal{G}_2, v_1), \dots (\mathcal{G}_\ell, v_1)$ be
    $\ell$ instance of \textsc{NS-TEXP} over the vertex set $V = \{v_1, v_2, \ldots, v_n\}$.
	
	The idea of the construction is to sequentially concatenate the $\ell$ instances one after
    the other as they share the same vertex set. Then, we introduce additional vertices that allow us to select and enter exactly one instance such that we can visit all the vertices in $V$ if and only if the selected instance is a yes-instance. Thereby, the additional vertices ensure that we only visit vertices in $V$ in one
    of the instances as we would otherwise fail to visit all of the additional vertices.

    Now, we describe the gadget implementing the selection of the instance in $\mathcal{G}$.
    We first present the idea to force the choice between two instances $(\mathcal{G}_1,v_1)$
    and $(\mathcal{G}_2, v_1)$ using $2$ control vertices and a copy $V'$ of $V$.
    Let $V' = \{v_1,', v_2', \ldots, v_n'\}$
    and $x_1, y_2 \notin V$ be additional vertices.
    Let $G_1^1, G_2^1, \ldots, G_{L_1}^1$ be the sequence of snapshots of $\mathcal{G}_1$
    and $G_1^2, G_2^2, \ldots, G_{L_2}^2$ be the sequence of snapshots of $\mathcal{G}_2$.
    We define $\mathcal{G}$ as follow, $V(\mathcal{G}) = V \cup V' \cup \{x_1, y_1\}$,
    $L = L_1 + L_2 + 8$, and $G_1, G_2, \ldots, G_L$ is the sequence of snapshots of
    $\mathcal{G}$, where
    on the first $L_1+4$ snapshots, we have $E(G_1) = \{\{x_1, v_1\}\}$, $E(G_2) = \{\{x_1, y_1\}\}$,
     $E(G_{i+2}) = E(G_i^1) \text{ for all $i \in [L_1]$}$, $E(G_{L_1 + 3}) = \{\{v_j, v'_j\} \mid j \in [|V|]\}$, and $E(G_{L_1 + 4}) = \{\{v'_j, x_1\}\mid j \in [|V|]\}$.
     Then, for the remaining $L_2+4$ snapshots, we repeat the same sequence with $\mathcal{G}_2$ instead of $\mathcal{G}_1$.
     Therefore we have, $E(G_{L_1 + 5}) = E(G_1) = \{\{x_1, v_1\}\}$, $E(G_{L_1 + 6}) = E(G_2) = \{\{x_1, y_1\}\}$, $E(G_{L_1 + 6 + i}) = E(G_i^2) \text{ for all $i \in [L_2]$}$, $E(G_{L_1 + L_2 + 7}) = E(G_{L_1 + 3}) = \{\{v_j, v'_j\} \mid j \in [|V|]\}$, and $E(G_{L_1 + L_2 + 8}) = E(G_{L_1 + 4}) = \{\{v'_j, x_1\}\mid j \in [|V|]\}$.
    See \autoref{fig:first-level} for a drawing of the gadget.
    To complete the construction of $(\mathcal{G},v)$ as an instance of \textsc{NS-TEXP},
    we set the start vertex to be $v = x_1$.

    Intuitively, an exploration in $(\mathcal{G},x_1)$ works as follow. We start
    in $x_1$. Then, in snapshot~$G_1$, there are two options, we can move on to 
    $v_1$ and visit $V$ by solving the instance $\mathcal{G}_1$ or we stay in $x_1$ and visit
    $y_1$ in $G_2$. In both case, after the snapshots corresponding to $\mathcal{G}_1$
    in $\mathcal{G}$, the we must stay in $x_1$ in the snapshot $G_{L_1 + 4}$.
    For the second half of the sequence, we do the complement of the first half. If we visited $V$ through $\mathcal{G}_1$, then we
    stay in $x_1$ to visit $y_1$, otherwise, the we go to $v_1$ and explore $V$ by solving $\mathcal{G}_2$.

    Let us now formally prove that $(\mathcal{G}, x_1)$ is a yes-instance
    if and only if either $(\mathcal{G}_1, v_1)$ or $(\mathcal{G}_2, v_1)$ is a yes-instance.
    
    We begin with the if direction. Assume $(\mathcal{G}_1, v_1)$ or $(\mathcal{G}_2, v_1)$ is a yes-instance.
    Recall that $V'\cup \{x_1\}$ form a connected component in $G_{L_1 + 4} = G_{L_1 + L_2 + 8}$.
    Let $W'$ be a $(v'_1, v'_n)$-walk on $G_{L_1 + 4} = G_{L_1 + L_2 + 8}$ such that
    $V(W') = V'$.
    Assume $(\mathcal{G}_1, v_1)$ is a yes-instance. By \autoref{obs:equiv},
    there exists a monotone $(v_1, u)$-walk $W_1$ for some $u \in V$ such that $V(W_1) = V$.
    From $W_1$, we build a monotone $(x_1, y_1)$-walk $W$ that visits all vertices in
    $V \cup V' \cup \{x_1, y_1\}$ as follows: $W = x_1, x_1v_1, W_1, uu', u', u'v'_1, W', v'_n x_1, x_1, x_1y_1,y_1$. It
    is straightforward to see that $W$ is a monotone walk and $V(W) = V \cup V' \cup \{x_1, y_1\}$. Therefore,
    $(\mathcal{G}, x_1)$ is a yes-instance.
    Now, assume $(\mathcal{G}_2, v_1)$ is a yes-instance. By \autoref{obs:equiv},
    there exists a monotone $(v_1, u)$-walk $W_2$ for some $u \in V$ such that $V(W_2) = V$.
    From $W_2$, we build a monotone $(x_1, v'_n)$-walk $W$ that visits all vertices in
    $V \cup V' \cup \{x_1, y_1\}$ as follows: $W = x_1, x_1y_1, y_1, y_1x_1, x_1, x_1v_1, W_1, uu', u', u'v'_1, W'$. It is
    straightforward to see that $W$ is a monotone walk and $V(W) = V \cup V' \cup \{x_1, y_1\}$. Therefore,
    $(\mathcal{G}, x_1)$ is a yes-instance.

    We now continue with the only if direction. Assume $(\mathcal{G}, x_1)$ is a yes-instance.
    By \autoref{obs:equiv}, there exists a monotone $(x_1, u)$-walk $W$ for some $u \in V \cup V' \cup \{x_1, y_1\}$
    such that $V(W) = V \cup V' \cup \{x_1, y_1\}$.
    The walk $W$ has to visit $y_1$ either in the snapshot $G_2$ or in the snapshot
    $G_{L_1 + 6}$. Assume that the walk visit $y_1$ in snapshot $G_2$, then it is easy
    to see that the walk has to return to $x_1$ in $G_2$. From there, the walk must go to $v_1$ in $G_{L_1 + 5}$ to be able to reach the vertices in $V$.
    Between time step $L_1 + 7$ and $L_1 + L_2 + 6$, $\mathcal{G}$ behaves like
    $\mathcal{G}_2$ on~$V$. Therefore, the vertices in $V$ are visited by a subwalk $W_2$
    of $W$. By construction of $\mathcal{G}$, $W_2$ is a monotone walk on $\mathcal{G}_2$
    starting at $v_1$. Therefore, $(\mathcal{G}_2, v_1)$ is a yes-instance.
    Symmetrically, if the walk visit $y_1$ in $G_{L_1+6}$, then the walk has to visit
    $V$ in a subwalk $W_1$ between time step $3$ and $L_1+3$. The walk $W_1$
    correspond to a monotone walk of $\mathcal{G}_1$. Therefore $(\mathcal{G}_1, v_1)$
    is a yes-instance.

    In fact, the gadget from \autoref{fig:first-level} can be recursively nested to select among more choices. For this, we see the dotted squares in \autoref{fig:first-level} containing the vertices from $V \cup V'$ as two black boxes with the property, that we can enter the black boxes over the vertex $x_1$ in the snapshots $\mathcal{G}_1$ and $\mathcal{G}_{L_1 +5}$, respectively, traverse the black boxes and exit them again to the vertex $x_1$ in the snapshots $\mathcal{G}_{L_1+4}$ and $\mathcal{G}_{L_1 +L_2 + 8}$, respectively. The crucial observation is that an exploration can exclusively only visit the first black box or the second black box but not both, as otherwise the vertex $y_1$ will not be visited. We can use this property to nest the gadget into itself in order to select over more than two instances.
    \autoref{fig:second-level} shows how to nest the gadget into itself by inserting the gadget as the black box and thereby obtaining a two layered gadget that selects among four different choices. For four instances, the boxes in \autoref{fig:second-level} are realized with the dotted boxes from \autoref{fig:first-level}. For eight instances, the boxes are again realized with the gadget from \autoref{fig:first-level}. It is easy to see how the recursive nesting continues. 
    The recursive nesting introduces $\lceil 2\log(\ell) \rceil$ control vertices for the selection of the instances and $|V|$ control variables as the copy $V'$ on the highest layer.
    By the recursive nesting of this gadget, we keep the property that on each layer, an exploration can only visit exactly one of the black boxes, and hence, in total, an exploration can only visit exactly one of the instances $(\mathcal{G}_i, v_1)$ for
    $1 \leq i \leq \ell$, as otherwise one of the control variables $y_j$ for $1 \leq j \leq \lceil 2\log(\ell) \rceil$, would not be visited. 
    By similar arguments as in the case of two instances, we get that if there exists an exploration for some instance $(\mathcal{G}_i, v^i)$, $1 \leq i \leq \ell$,
    then there exists an exploration for
    $(\mathcal{G}, v)$ and vice versa, if there exists an exploration for $(\mathcal{G}, v)$, then this exploration visits the vertices in $V$ by entering exactly one of the instances $(\mathcal{G}_i, v^i)$, $1 \leq i \leq \ell$.
    Therefore, $(\mathcal{G}, v)$ is
    a yes-instance if and only if there exists some $1 \leq i \leq \ell$ such that
    $(\mathcal{G}_i, v^i)$ is a yes-instance.
\end{proof}
\begin{figure}
    \centering
    \begin{tikzpicture}
        \def\v#1#2{
            (-0.2, #2) node {$v_#1$}
            (0, #2) edge (1.5, #2)
            (3.5, #2) edge (7.5, #2)
            (9.5, #2) edge (12, #2)
        }
        \def\vprime#1#2{
            (-0.2, #2) node {$v_#1'$}
            (0, #2) edge (12, #2)
        }
        \draw[draw=black] (1.5, 2) rectangle (3.5, 3.5);
        \draw[draw=black] (7.5, 2) rectangle (9.5, 3.5);
        \draw[thick, dashed] (0.25,1.8) rectangle (5.75,5.75);
        \draw[dashed] (6.25,1.8) rectangle (11.75,5.75);
        \node at (-0.2, 0.5) {\scriptsize{$t$}};
        \node at (0.5, 0.5) {\scriptsize{$1$}};
        \node at (1, 0.5) {\scriptsize{$2$}};
        \node at (1.5, 0.5) {\scriptsize{$3$}};
        \node at (2.5, 0.5) {$\cdots$};
        \node at (4.25, 0.5) {\scriptsize{$L_1{+}3$}};
        \node at (5, 0.5) {\scriptsize{$L_1{+}4$}};
        \node at (5.75, 0.5) {$\cdots$};

        \node at (-0.2, 1.5) {$x_1$};
        \node at (-0.2, 1) {$y_1$};

        \path[thick,dotted]
            (0,1) edge (12,1)
            (0,1.5) edge (12,1.5)

            \v{1}{2.1}
            \v{2}{2.4}
            (-0.2, 3) node {$\vdots$}
            \v{n}{3.4}

            \vprime{1}{4.0}
            \vprime{2}{4.5}
            (-0.2, 5.1) node {$\vdots$}
            \vprime{n}{5.5}
            ;
        \draw [thick, decorate, decoration = {calligraphic brace,raise=10pt}] (-0.2, 2.0) --  +(0,1.5) node[pos=0.5, left=10pt]{$V$};
        \draw [thick, decorate, decoration = {calligraphic brace,raise=10pt}] (-0.2, 3.9) --  +(0,1.5) node[pos=0.5, left=10pt]{$V'$};

        \node[circle,fill=black,inner sep=0pt,minimum size=6pt] (x1_1) at (0.5,1.5) {};
        \node[circle,fill=black,inner sep=0pt,minimum size=6pt] (x1_2) at (1,1.5) {};
        \node[circle,fill=black,inner sep=0pt,minimum size=6pt] (x1_6) at (5,1.5) {};
        \node[circle,fill=black,inner sep=0pt,minimum size=6pt] (x1_7) at (6.5,1.5) {};
        \node[circle,fill=black,inner sep=0pt,minimum size=6pt] (x1_8) at (7,1.5) {};
        \node[circle,fill=black,inner sep=0pt,minimum size=6pt] (x1_10) at (11,1.5) {};

        \node[circle,fill=black,inner sep=0pt,minimum size=6pt] (y1_2) at (1,1) {};
        \node[circle,fill=black,inner sep=0pt,minimum size=6pt] (y1_8) at (7,1) {};

        \node[circle,fill=black,inner sep=0pt,minimum size=6pt] (v1_1) at (0.5,2.1) {};
        \node[circle,fill=black,inner sep=0pt,minimum size=6pt] (v1_5) at (4.25,2.1) {};
        \node[circle,fill=black,inner sep=0pt,minimum size=6pt] (v1_7) at (6.5,2.1) {};
        \node[circle,fill=black,inner sep=0pt,minimum size=6pt] (v1_9) at (10.25,2.1) {};

        \node[circle,fill=black,inner sep=0pt,minimum size=6pt] (v2_5) at (4.25,2.4) {};
        \node[circle,fill=black,inner sep=0pt,minimum size=6pt] (v2_9) at (10.25,2.4) {};

        \node[circle,fill=black,inner sep=0pt,minimum size=6pt] (vn_5) at (4.25,3.4) {};
        \node[circle,fill=black,inner sep=0pt,minimum size=6pt] (vn_9) at (10.25,3.4) {};

        \node[circle,fill=black,inner sep=0pt,minimum size=6pt] (v1'_5) at (4.25,4) {};
        \node[circle,fill=black,inner sep=0pt,minimum size=6pt] (v1'_6) at (5,4) {};
        \node[circle,fill=black,inner sep=0pt,minimum size=6pt] (v1'_9) at (10.25,4) {};
        \node[circle,fill=black,inner sep=0pt,minimum size=6pt] (v1'_10) at (11,4) {};

        \node[circle,fill=black,inner sep=0pt,minimum size=6pt] (v2'_5) at (4.25,4.5) {};
        \node[circle,fill=black,inner sep=0pt,minimum size=6pt] (v2'_6) at (5,4.5) {};
        \node[circle,fill=black,inner sep=0pt,minimum size=6pt] (v2'_9) at (10.25,4.5) {};
        \node[circle,fill=black,inner sep=0pt,minimum size=6pt] (v2'_10) at (11,4.5) {};

        \node[circle,fill=black,inner sep=0pt,minimum size=6pt] (vn'_5) at (4.25,5.5) {};
        \node[circle,fill=black,inner sep=0pt,minimum size=6pt] (vn'_6) at (5,5.5) {};
        \node[circle,fill=black,inner sep=0pt,minimum size=6pt] (vn'_9) at (10.25,5.5) {};
        \node[circle,fill=black,inner sep=0pt,minimum size=6pt] (vn'_10) at (11,5.5) {};

        \path[thick] (x1_1) edge (v1_1)
            (x1_2) edge (y1_2)
            (v1_5) edge[bend left] (v1'_5)
            (v2_5) edge[bend left] (v2'_5)
            (vn_5) edge[bend left] (vn'_5)
            (v1'_6) edge (x1_6)
            (v2'_6) edge[bend left] (x1_6)
            (vn'_6) edge[bend left] (x1_6)
            (x1_7) edge (v1_7)
            (x1_8) edge (y1_8)
            (v1_9) edge[bend left] (v1'_9)
            (v2_9) edge[bend left] (v2'_9)
            (vn_9) edge[bend left] (vn'_9)
            (v1'_10) edge (x1_10)
            (v2'_10) edge[bend left] (x1_10)
            (vn'_10) edge[bend left] (x1_10)
            ;

        \node at (3,6) {$x_1 = 1$};
        \node at (2.5,2.75) {$\mathcal{G}_1$};
        \node at (9,6) {$x_1 = 0$};
        \node at (8.5,2.75) {$\mathcal{G}_2$};

        \path[dashed] (6,0.25) edge +(0,6);

    \end{tikzpicture}
    \caption{Gadget of \autoref{thm:no-kernel-n}:
        Given two instance of \textsc{NS-EXP}, $(\mathcal{G}_1, v_1)$ and $(\mathcal{G}_2, v_1)$,
        this construction shows how to connect them to create an instance $(\mathcal{G}, x_1)$ such
        that $(\mathcal{G}, x_1)$ is a yes-instance if and only if $(\mathcal{G}_1, v_1)$ or
        $(\mathcal{G}_2, v_1)$ is a yes-instance.
    }%
    \label{fig:first-level}
\end{figure}
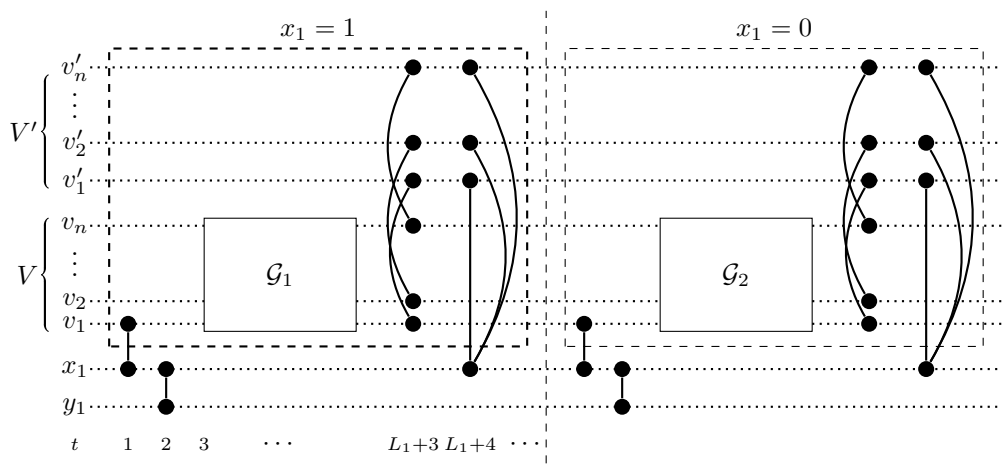

\begin{figure}
	\centering
	\scalebox{0.65}{%
	\begin{tikzpicture}
		\draw[draw=black] (1, 2) rectangle (4, 3); 
		\draw[draw=black] (5.5, 2) rectangle (8.5, 3); 
		
		\node at (-0.5, 1) {$x_2$};
		\node at (-0.5, 0) {$y_2$};
		\path[thick, dotted] (-0.3,0) edge (20.3,0)
		(-0.3,1) edge (20.3,1);
		
		\node at (-0.5, -1) {$x_1$};
		\node at (-0.5, -2) {$y_1$};
		\path[thick, dotted] (-0.3,-2) edge (20.3,-2)
		(-0.3,-1) edge (20.3,-1);
		
		\node[circle,fill=black,inner sep=0pt,minimum size=6pt] (1) at (0.8,1) {};
		\node[circle,fill=black,inner sep=0pt,minimum size=6pt] (2) at (1.5,1) {};
		\node[circle,fill=black,inner sep=0pt,minimum size=6pt] (3) at (1.5,0) {};
		
		\node[circle,fill=black,inner sep=0pt,minimum size=6pt] (4) at (4.2,1) {};
		\node[circle,fill=black,inner sep=0pt,minimum size=6pt] (5) at (5.3,1) {};
		\node[circle,fill=black,inner sep=0pt,minimum size=6pt] (6) at (6,1) {};
		\node[circle,fill=black,inner sep=0pt,minimum size=6pt] (7) at (6,0) {};
		
		\node[circle,fill=black,inner sep=0pt,minimum size=6pt] (8) at (8.7,1) {};
		
		\node[circle,fill=black,inner sep=0pt,minimum size=6pt] (21) at (0.2,-1) {};
		\node[circle,fill=black,inner sep=0pt,minimum size=6pt] (22) at (0.2,1) {};
		\node[circle,fill=black,inner sep=0pt,minimum size=6pt] (23) at (0.7,-1) {};
		\node[circle,fill=black,inner sep=0pt,minimum size=6pt] (24) at (0.7,-2) {};
		\node[circle,fill=black,inner sep=0pt,minimum size=6pt] (25) at (9.3,-1) {};
		\node[circle,fill=black,inner sep=0pt,minimum size=6pt] (26) at (9.3,1) {};
		
		\node[circle,fill=black,inner sep=0pt,minimum size=6pt] (27) at (10.7,-1) {};
		\node[circle,fill=black,inner sep=0pt,minimum size=6pt] (28) at (10.7,1) {};
		\node[circle,fill=black,inner sep=0pt,minimum size=6pt] (29) at (11.2,-1) {};
		\node[circle,fill=black,inner sep=0pt,minimum size=6pt] (30) at (11.2,-2) {};
		
		\node[circle,fill=black,inner sep=0pt,minimum size=6pt] (31) at (19.8,-1) {};
		\node[circle,fill=black,inner sep=0pt,minimum size=6pt] (32) at (19.8,1) {};

		\path[thick] (1) edge[bend left] (1,2.5)
		(2) edge (3)
		(4,2.5) edge[bend left] (4)
		(5) edge[bend left] (5.5,2.5)
		(6) edge (7)
		(8.5,2.5) edge[bend left] (8);

		\path[thick] (21) edge[bend left] (22)
		(23) edge (24)
		(26) edge[bend left] (25)
		(27) edge[bend left] (28)
		(29) edge (30)
		(32) edge[bend left] (31);
		
		\node at (2.5,2.5) {$x_1 = 1,\ x_2=1$};
		\node at (7,2.5) {$x_1 = 1,\ x_2=0$};
		
		\path[dashed] 
		(4.75,-0.5) edge (4.75,3.5)		
		(10,-2.5) edge (10, 3.5);
		
		
		\draw[draw=black] (11.5, 2) rectangle (14.5, 3); 
		\draw[draw=black] (16, 2) rectangle (19, 3);

		\node[circle,fill=black,inner sep=0pt,minimum size=6pt] (11) at (11.3,1) {};
		\node[circle,fill=black,inner sep=0pt,minimum size=6pt] (12) at (12,1) {};
		\node[circle,fill=black,inner sep=0pt,minimum size=6pt] (13) at (12,0) {};
		
		\node[circle,fill=black,inner sep=0pt,minimum size=6pt] (14) at (14.7,1) {};
		\node[circle,fill=black,inner sep=0pt,minimum size=6pt] (15) at (15.8,1) {};
		\node[circle,fill=black,inner sep=0pt,minimum size=6pt] (16) at (16.5,1) {};
		\node[circle,fill=black,inner sep=0pt,minimum size=6pt] (17) at (16.5,0) {};
		
		\node[circle,fill=black,inner sep=0pt,minimum size=6pt] (18) at (19.2,1) {};
		
		\path[thick] (11) edge[bend left] (11.5,2.5)
		(12) edge (13)
		(14.5,2.5) edge[bend left] (14)
		(15) edge[bend left] (16,2.5)
		(16) edge (17)
		(19,2.5) edge[bend left] (18);
		
		\node at (13,2.5) {$x_1 = 0,\ x_2=1$};
		\node at (17.5,2.5) {$x_1 = 0,\ x_2=0$};
		
		\path[dashed] 
		(15.25,-0.5) edge (15.25,3.5);
		
		\draw[draw=black,dotted] (0.5, -0.5) rectangle (9, 3.5); 
		\draw[draw=black,dotted] (11, -0.5) rectangle (19.5, 3.5); 
	\end{tikzpicture}}
    \caption{Gadget for \autoref{thm:no-kernel-n}: This shows how to nest the gadget of \autoref{fig:first-level} into a two layered gadget  that selects between four choices.
    }%
	\label{fig:second-level}
\end{figure}
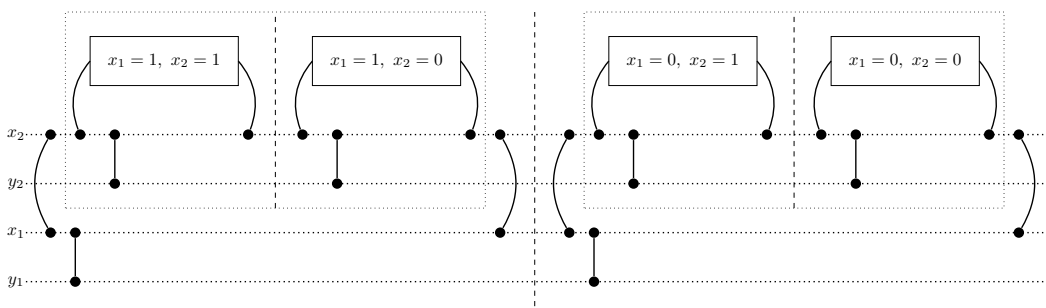
\end{toappendix}

\subsection*{Lower Bounds for \textsc{$k$-arb NS-TEXP}}
In~\cite{ErlebachS22Journal} an \FPT-algorithm based on color coding is given for the problem \textsc{$k$-arb NS-TEXP} with parameter $k$. We show in the following that \textsc{$k$-arb NS-TEXP} parameterized by $L$ is \W[1]-hard solving an open problem from~\cite{ErlebachS22Journal}. This hardness contrasts the \FPT-algorithm in parameter $L$ for the problem \textsc{NS-TEXP}. 
We further show that \textsc{$k$-arb NS-TEXP} does not admit a polynomial kernel in $L+k$ unless $\coNP \subseteq \NP/\poly$. This result is based on a cross-composition from the problem \textsc{$k$-arb NS-TEXP} to itself. 
\begin{theorem}
	\label{thm:no-FPT-L}
	\textsc{$k$-arb NS-TEXP} is \W[1]-hard parameterized by $L$.
\end{theorem}
We show the claim by a reduction from \textsc{Multicolored Independent Set}.
\begin{definition}[\textsc{Multicolored Independent Set}]
	\ \\
	Input: $k$-partite graph $G = (V_1, V_2, \dots, V_k, E)$.\\
	Question: Does there exist a set of vertices $S = \{v_1, v_2, \dots, v_k\}$ such that $v_i \in V_i$ for $1\leq i \leq k$ and $S$ is an independent set in $G$?
\end{definition}
It was shown in~\cite{CyganFKLMPPS15} that \textsc{Multicolored Independent Set} is \W[1]-hard parameterized by~$k$, even on $d$-regular graphs (graphs where each vertex has degree $d$) and $d$ is fixed.
\begin{proof}
	Let $G=(V_1, V_2, \dots, V_k, E)$ be a $d$-regular graph and an instance of \textsc{Multicolored Independent Set}. 
	First, we subdivide each edge in $E$ by introducing a new set of vertices~$V_E$. Let $G'=(V_1, V_2, \dots, V_k, V_E, E')$ be the graph obtained after subdividing each edge. Observe that in $G'$ each vertex in $\bigcup_{1 \leq i \leq k}V_i$ only has neighbors in $V_E$. Further, as $G$ is $d$-regular, for any set $H=\{v_1, v_2, \dots v_k \mid v_i \in V_i, 1\leq i\leq k\}$ it holds that $|\bigcup_{v \in H}N_{G'}(v)| = kd$ if and only if $H$ is an independent set in $G$. To see this, observe that if two vertices $u, v$ in $H$ are adjacent in $G$, then they share a neighbor in $G'$. As each set $N_{G'}(v)$ is of size $d$ for $v \in \bigcup_{1 \leq i \leq k}V_i$ the equation $|\bigcup_{v_\in H}N_{G'}(v)| = kd$ with $|H| = k$ only holds if the sets $N_{G'}(v)$ are disjunctive.
	We further note that each vertex $v\in V_E$ is of degree two and, as $G$ is a $k$-partite graph, the two neighbors of $v$ are in different sets $V_i \neq V_j$. 
	
	We will now build a non-strict temporal graph $\mathcal{G}$ from $G'$ based on the following idea. In time step $t = (i-1)\cdot 2$ we can choose a vertex from $V_i$ and visit all its neighbors in $V_E$ in the next step. If the $k$ vertices that we select this way form an independent set, then we are able to visit $kd$ vertices from $V_E$, otherwise we will see a vertex from $V_E$ twice. 
	
	We now describe the construction more formally.
    Similar to \autoref{thm:const-c}, we are interesting in the vertices of the connected components and not in theirs structures. Thereby, we define the connected components as sets of vertices and assume that each connected component forms a clique.
	For $t=0$, we let $C_{0,1} = V_1$ be a connected component and each other vertex is isolated. Further, we set $v \in C_{0,1}$ as the start vertex of the sought temporal walk.
	Now consider $t=1$. Let $V_1 = \{v^1_1, v^1_2, \dots, v^1_{\ell_1}\}$. Then, for $1 \leq j \leq \ell_1$, we let $C_{1,j} = \{v^1_j\} \cup N_{G'}(v^1_j)$ be a connected component. Each other vertex is isolated.
	Now, let $t=(i-1)\cdot 2$ for $2 \leq i \leq k$. 
	Then, we set $C_{t,1} = V_i \cup V_{i-1}$ as a connected component and each other vertex is isolated.
	For the time step $t+1$, let $V_{i} = \{v^i_1, v^i_2, \dots, v^i_{\ell_i}\}$. Then, for $1 \leq j \leq \ell_i$, we set $C_{1,j} = \{v^i_j\} \cup N_{G'}(v^i_j)$ as a connected component and each other vertex is isolated.
	
	Observe that $L = 2k$. It remains to set the parameter $k'$ being the number of vertices to be visit on a temporal walk in $\mathcal{G}$.
    We set $k' = |\bigcup_{1 \leq i \leq k}V_i| + kd$.
	
	Now assume, there exists a multicolored independent set $H$ for $G$. For, $1\leq i \leq k$, let $v_i \in H$ be the vertex contained in $V_i$. Then, following the vertex $v_i$ in the time steps $2(i-1)$ and $2(i-1)+1$ allows us to visit all vertices in $N_{G'}(H)$ which are $kd$ many by the assumption that $H$ is an independent set. As the vertices in $V_i$ and $V_{i+1}$ share a component in the time step $2i$ we are able to change following vertex $v_i$ to following vertex $v_{i+1}$ in time step $2i$ and further visit all vertices in $\bigcup_{1 \leq i \leq k}V_i$.
	
	For the other direction, assume there exists a monotone walk $W$ in $\mathcal{G}$ that visits at least $k' = |\bigcup_{1 \leq i \leq k}V_i| + dk$ many vertices.
	Recall that each vertex in $V_E$ has at most two neighbors and no two neighbors are contained in the same set $V_i$.
	Hence, in each odd time step, we can visit at most $d$ new vertices from $V_E$. As in each even time step, the vertices in $V_E$ are contained in singleton components, we can only visit the vertices in $V_E$ in odd time steps. 
	We claim that (*) in each even time step $t$, $W$ visits the first component. This way, $W$ visits all vertices in $\bigcup_{1 \leq i \leq k}V_i$ and needs to visit only $kd$ vertices from $V_E$. 
	Assume (*) is not true. Then, $W$ is in a singleton component in time step $t$, and, by construction, stays in a singleton component in time step $t+1$. But then, $W$ does not visit any vertex from $V_E$ in time step $t+1$ and as we can only visit at most $d$ new vertices from $V_E$ in an odd time step (and none in an even time step), it is no longer possible to visit $dk$ vertices from $V_E$ during the walk.
	Now, collecting the vertices from $\bigcup_{1 \leq i \leq k}V_i$ that $W$ visits in odd time steps gives us a multicolored independent set for $G$.
\end{proof}

\begin{theorem}%
	\label{thm:no-kernel-k-L}
	Unless $\coNP \subseteq \NP/\poly$ \textsc{$k$-arb NS-TEXP} does not admit a polynomial kernel parameterized by $k+L$.
\end{theorem}
\begin{toappendix}
\begin{proof}
    We show the result using \autoref{crosscollapse} via a cross-composition from \textsc{$k$-arb NS-TEXP} into itself.
	Therefore, we reduce $\ell$ instances $(\mathcal{G}_1,v^1, k_1), (\mathcal{G}_2, v^2, k_2), \dots, (\mathcal{G}_\ell, v^\ell, k_\ell)$ of \textsc{$k$-arb NS-TEXP} into one instance $(\mathcal{G}, v, k)$ of \textsc{$k$-arb NS-TEXP} such that $(\mathcal{G}, v, k)$ is a yes instance if and only if at least one of the instances $(\mathcal{G}_i, v^i, k_i)$ for $1 \leq i \leq \ell$ is a yes-instance.  Here, $\mathcal{G}_i$ is a temporal graph, $v^i$ is the start vertex, and $k_i$ is the number of vertices to visit in a temporal walk in $\mathcal{G}_i$.
	
	Without loss of generality, we can assume that $\ell$ is of the form $2^x$ and
    all the vertex set $V(\mathcal{G}_i)$ are pairwise disjoint. 
	Further, we can assume $k_1 = k_2 = \dots = k_\ell$. Otherwise, let $k' = \max(k_1, \dots, k_\ell)$. Then, for the $i$th instance, we can add $k' - k_i$ new vertices adjacent to $v_i$ in all snapshots. Thereby we ensure to not change the structure of the instance and reach $k' - k_i$ more vertices. Hence, we can set $k_i := k'$.
	We can further assume that the lifetimes $L_1, L_2, \dots, L_\ell$ of the $\ell$ many temporal graphs are identical as otherwise we can pad graphs with shorter lifetimes by duplicating the last snapshot for every missing time steps. Hence, from now on, we refer to the lifetime $L_1 = L_2 = \dots = L_\ell$ as $L'$ and to the parameter $k_1 = k_2=  \dots = k_\ell$ as~$k'$.
	
	We now construct from $(\mathcal{G}_1,v^1, k_1), (\mathcal{G}_2, v^2, k_2), \dots (\mathcal{G}_\ell, v^\ell, k_\ell)$
	an instance $(\mathcal{G},v, k)$ with $k = k'+2\log(\ell)-1$ and $L = L'+\log(\ell)$.

    We define $V(\mathcal{G}) =\{v_1, v_2, \ldots, v_{2\ell-1}\} \uplus \biguplus_{1\leq i \leq \ell} V(\mathcal{G}_i)$. We then construct the last $L$ snapshot of $\mathcal{G}$ as the union
    of the snapshots in ${(\mathcal{G}_i)}_{1 \leq i \leq \ell}$. Formally, for
    $1 \leq t \leq L$,
    we let $E(\mathcal{G}(t + \log(\ell))) = \biguplus_{1 \leq i \leq \ell} E(\mathcal{G}_i(t))$.
    To define the first $\log(\ell)$ time steps, we use the $2\ell-1$ new vertices $v_1, v_2, \dots v_{2\ell-1}$ to $V(\mathcal{G})$. We arrange the new vertices in a balanced binary tree such that the root of the tree is adjacent to its two children in time step $t=0$ and in each time step $t$ for $1 \leq t \leq \ell-1$, the nodes of the tree in level $t$ are adjacent with their two children.
	In the time step $t=\ell$ the $\ell$ leafs $l_1, l_2, \dots, l_\ell$ of the tree are adjacent to the respective start vertices of the $\ell$ many instances. More precisely, for the time step $t= \ell$, the vertex $l_i$ is adjacent to the start vertex $v^i$ of the $i$th instance.
	For the remaining cases not specified above, each of the newly introduced vertices $v_1, v_2, \dots, v_{2\ell-1}$ is an isolated vertex.
	
	Intuitively, the binary tree added in front of the union of the $\ell$ many temporal graphs allows us to select one of those graphs. During this selection, we visit $2\log(\ell)-1$ many of the newly introduced vertices until we reach an instance $\mathcal{G}_i$ in time step $\ell$. From time step $\ell$ onward, a temporal walk in $\mathcal{G}$ equals a temporal walk in some $\mathcal{G}_i$ dependent on the selection $i$ made by the walk in the binary tree at the beginning. Thereby, it is clear that we can visit $k_j+2\log(\ell)-1$ many vertices on a temporal walk in $\mathcal{G}$ if and only if we can visit $k_j$ many vertices on a temporal graph in some graph $\mathcal{G}_j$.	
\end{proof}
\end{toappendix}

\input{kernel}

\section*{Conclusion}
We initialized the study of polynomial kernels for exploration problems on temporal graphs. We showed that for the problems \textsc{NS-TEXP}, \textsc{$k$-arb NS-TEXP}, and \textsc{Weighted k-arb-NS-TEXP}, unless  $\NP \subseteq \coNP/\poly$, there does not exist a polynomial kernel for the parameters 
number of vertices $n$, lifetime $L$, and number of vertices to visit $k$; and for the combined parameters $L+k$, $L + \gamma$, and $k+\gamma$, where $\gamma$ is the maximal number of connected components per time step. 
In fact, by a straight forward reduction that repeats each snapshot sufficiently many times, all of our hardness results, that do not involve the parameter $L$, carry over to the strict settings \textsc{TEXP} and \textsc{$k$-arb TEXP} where a temporal exploration traverses at most one edge per time step.
We showed that the temporal exploration problems remain \NP-hard restricted to temporal graphs where the underlying graph is a tree of depth two. 
From a parameterized complexity point of view, this eliminates most of the common structural parameters considered on the underlying graph. Nonetheless, we were able to identify a new parameter of a temporal graph 
$p(\mathcal{G}) = \sum_{i=1}^{L} (|E(G_i)|) - |V(G)| +1$ that captures how close the temporal graph is to a tree where each edge appears exactly once.
Our parameter can also be seen as a notion of sparsity for temporal graphs.
For this parameter~$p(\mathcal{G})$, we were able to obtain a polynomial kernel for \textsc{Weighted $k$-arb NS-TEXP}.
Using simplified reduction rules, we can obtain a kernel of linear size for \textsc{NS-TEXP}. While the reduction from \textsc{NS-TEXP} to \textsc{TEXP} blows up the parameter $p(\mathcal{G})$, our reduction rules can still be adapted to obtain polynomial kernels for the strict variants of the considered exploration problems. The natural next step would be to evaluate how useful our parameter is for other problems on temporal graphs.

\bibliography{bib}
\end{document}

%% file: kernel.tex

\section{Polynomial Kernel for Bounded Number of Edge Appearances}\label{sec:kern}
In the field of static graphs, trees are nice structures that often lead to efficient algorithms.
Building on this fact, structural parameters, like treewidth, defining a distance to trees were
successfully introduced and used to design efficient algorithms. Unfortunately, we have seen
in \autoref{thm:treeNP} that even when the underlying graph is a tree with vertex
cover number $2$ the problems we consider remain \NP-hard.
Therefore, in this section, we introduce a new structural parameter $p(\mathcal{G})$
which, intuitively, characterises how far $\mathcal{G}$
is from being a tree in which each edge appears exactly once.
Using this new parameter, we investigate the complexity of \textsc{$k$-arb NS-TEXP}.

We restate the definition of $p(\mathcal{G})$. Recall that for a temporal
graph $\GL$, we write $\m = \sum_{i=1}^{L} |E(G_i)|$. Observe that if $\m<n-1$, then the underlying graph $G$ of $\mathcal{G}$ is disconnected. Let $X$ be the set of vertices of the connected component of $G$ containing the source vertex $v$. Then, \textsc{$k$-arb NS-TEXP} for $\G$ is equivalent to \textsc{$k$-arb NS-TEXP} for the temporal graph $\G[X]=(G_1[X],\ldots,G_L[X])$. Hence, we can assume that the underlying graph is connected and $\m\geq n-1$. This allows us to consider the above-guarantee parameter $p(\mathcal{G}) = \mathfrak{m} - n+1$. We show that   \textsc{$k$-arb NS-TEXP} admits a polynomial kernel in $p(\mathcal{G})$ when the underlying graph is connected. In fact, we give a kernel for a more general variant of the problem with vertex weights called \WEX.

The starting point of our kernelization algorithm is a polynomial time algorithm for
\WEX{} on temporal graph in which each edge appears exactly once and the underlying
graph is a tree (\autoref{lem:poly-trees}).
In addition to this algorithm, the crucial observation that leads to a kernel is that the underlying
graph of the input temporal graph has a feedback edge set of bounded size (\autoref{lem:fes}).
By combining this observation and the fact that there is a bounded number of edges
repetitions, we show that the input temporal graph $\mathcal{G}$ has some specific
structure. More specifically, we show that there exists a core set $X \subseteq V(\mathcal{G})$
of vertices of bounded size such that the underlying graph of the remaining graph $\mathcal{G} - X$ is a forest $\mathcal{F}$ with the
following property. In $\mathcal{G} - X$, each edge appears exactly once and each tree
of $\mathcal{F}$ is connected to $X$ by at most two edges (see \autoref{fig:struct} and \autoref{lem:struct}).
For each tree $T$ in $\mathcal{F}$, depending on its structure and its interaction with
$X$, we are able to describe all possible ways an exploration can visit some of the vertices of $T$.
Using the polynomial time algorithm for trees in which edges appear only once, we can then
compute the maximum weight contributed by the vertices in $T$ for each of those cases
and design a gadget of constant size that simulates the original tree $T$. Thereby, the gadget keeps the information of how many vertices, respectively vertex weights, were reachable in $T$ and is the reason why we need to work with the weighted version of \textsc{$k$-arb-NS-TEXP}.
This gives us several reduction rules (\autoref{rule:trees}-\ref{rule:contr}).
To conclude the kernel, we show that after applying these reduction rules exhaustively,
the obtained temporal graph has linear size in $p$ (\autoref{cl:size}).
Lastly, by using an algorithm proposed by Frank and Tardos (\autoref{prop:FT}), we
can bound the size of the obtained weights on the vertices by $\mathcal{O}(p^4)$.

Throughout this section we assume that considered temporal graphs have connected underlying graphs. Let $\GL$ be a temporal graph and $G$ its underlying graph.
For $e\in E(G)$, we denote by $A(e)=\{t\colon e\in E(G_t)\text{ for }1\leq t\leq L\}$ and set $m(e)=|A(e)|$. Note that $\m(\G)=\sum_{e\in E(G)}m(e)$.

By \autoref{obs:equiv}, an instance $(\G,w,v,k)$ of \WEX{} is a yes-instance if and only if the underlying graph $G$ has a monotone $(v,x)$-walk $W$ for some $x\in V(G)$ such that $w(V(W))\geq k$. Slightly abusing notation, we write 
$w(W)$ instead of $w(V(W))$ for the total weight of vertices visited by a walk.

We start by showing that the underlying graph has a feedback edge set of bounded size.
Recall that a set of edges $S\subseteq E(G)$ is a \emph{feedback edge set} of a graph $G$, if $G-S$ is a forest.
Let $\GL$ be a temporal graph and let  $p = \mathfrak{m} - n+1$. We denote by $R$ the set of edges of the underlying graph $G$ appearing at least twice in $\G$, that is, $e\in R$ if there are distinct $i,j\in[L]$ such that $e\in E(G_i)$ and $e\in E(G_j)$.
We refer to the edges of $R$ as \emph{red}. We set $B=E(G)\setminus R$ and call the edges of $B$ \emph{blue}.
Clearly, each blue edge occurs in $G_i$ for a unique $i\in[L]$.
We can make the following observation.

\begin{lemma}\label{lem:fes}
Let $\G$ be a temporal graph and let  $p = \mathfrak{m} - n+1$. Then, the underlying graph $G$ of $\G$ has a feedback edge set $S$ of size at most $p$ such that every red edge is in $S$.
\end{lemma}

\begin{proof}
Let $R$ be the set of red edges of $G$. Consider $q=\sum_{e\in R}m(e)-|R|\geq |R|$. 
Observe that $G$ has exactly $n-1+p-q$ edges. Therefore, $G$ has a feedback edge set $S'$ of size $p-q$. Then, $S=S'\cup R$ is a feedback edge set of $G$ of size at most $p$ such that $R\subseteq S$. 
\end{proof}
 
\noindent \autoref{lem:fes} implies that $G$ has a special structure that is used in our algorithm (see \autoref{fig:struct}).

\begin{figure}[ht]
\centering
\scalebox{0.7}{%
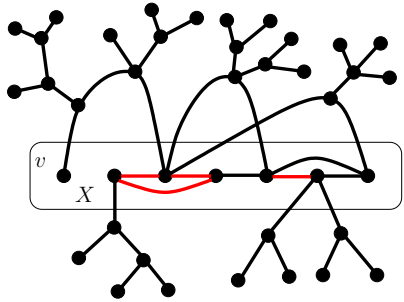}
\caption{Structure of the underlying graph $G$ of a temporal graph $\mathcal{G}$ when the
    parameter $p = \mathfrak{m} - n+1$ is bounded. Here $|X| \leq 4p$.
    (\autoref{lem:struct})
}%
\label{fig:struct}
\end{figure}

\begin{lemma}\label{lem:struct}
    Let $\G$ be a temporal graph with the underlying graph $G$ and let  $p = \mathfrak{m} - n+1$.  Let also $v\in V(G)$. Then, there is a set of vertices $X\subseteq V(G)$ of size at most $4p$ such that (i)~$v\in X$, (ii)~every red edge is in $G[X]$, (iii) $G-E(G[X])$ is a forest and (iv)~each connected component $F$ of $G-X$ is a tree with the properties that each vertex $x\in X$ has at most one neighbor in $F$ and 
\begin{enumerate}[(a)]
    \item\label{lem:struct:a} either $X$ has a unique vertex $x$ that has a neighbor in $F$,
    \item\label{lem:struct:b} or $X$ contain exactly two vertices $x$ and $y$ having neighbors in $F$. 
\end{enumerate}
Furthermore, if $q$ is the number of connected components $F$ of $G-X$ satisfying (\ref{lem:struct:b}), then $q\leq 4p-1$ and $G[X]$ has at most $5p-q-1$ edges.
\end{lemma}

\begin{proof}
By \autoref{lem:fes}, $G$ has a feedback edge set $S$ of size at most $p$ containing all red edges. We  define $Y=\{v\}\cup\{x\in V(G)\mid x\text{ is an endpoint of an edge of }S\}$. Note that $|Y|\leq 2p+1$.
Consider the graph $G'$ obtained from $G$ by the iterative deletion of vertices of degree one that are not in $Y$. Because $S$ is a feedback edge set, we obtain that $H=G'- S$ is a forest such that every vertex of degree at most one is in $Y$. It is a folklore knowledge that any tree with $\ell\geq 2$ vertices of degree one (leaves) has at most $\ell-2$ vertices of degree at least three. This implies that the set $Z$ of vertices of $H$ with degree at least three has size at most $2p-1$. We define $X=Y\cup Z$. By the construction, $|X|\leq 4p$ and (i)--(iv) are fulfilled.  

Let $q$ be the number of connected components $F$ of $G-X$ satisfying (\ref{lem:struct:b}).
Consider the multigraph $G''$ obtained from $G[X]$ by the following operation: for each connected component $F$ of $G-X$ satisfying (b), where $x$ and $y$ are the vertices of $X$ having neighbors in $F$, add the edge $xy$ to $G[X]$ and make it a multi-edge if $xy\in E(G[X])$.
Because of (iii), $G'' - E(G[X])$ is a tree on $|X|$ vertices, therefore, we have that $q\leq |X|-1\leq 4p-1$.
Because $G$ has a feedback edge set of size at most $p$, $G''$ has the same property. Therefore, $G''$ has at most $|X|-1+p\leq 5p-1$ edges. Hence, $G[X]$ has at most $5p-q-1$ edges. 
\end{proof}

By \autoref{cor:single}, \textsc{NS-TEXP} is \NP-complete even if every edge of the temporal graph appears only once. We show that \WEX can be solved in polynomial time if each edge appears exactly once and the underlying graph is a tree. If an edge $e$ appears exactly once, then we use $t(e)$ to denote the unique element of $A(e)$. 

\begin{lemma}\label{lem:poly-trees}
There is an algorithm running in $\Oh(nL)$ time that, given a temporal graph $\mathcal{T}=(F_1,\ldots,F_L)$ such that its underlying graph $T$ is a tree with  $m(e)=1$ for each $e\in E(T)$, a weight function $w\colon V(T)\rightarrow \mathbb{Z}_{>0}$ and two vertices $x,y\in V(T)$, either finds the maximum weight of a monotone $(x,y)$-walk $W$ in $T$ 
or reports that such a walk does not exist.
\end{lemma}
\begin{toappendix}
\begin{proof}
Let  $\mathcal{T}=(F_1,\ldots,F_L)$ be a temporal graph such that its underlying graph $T$ is a tree with  $m(e)=1$ for each $e\in E(T)$. Let $x,y\in V(T)$ and $w\colon V(T)\rightarrow \mathbb{Z}_{>0}$ be a weight function. 

If $y=x$, that is, $W$ is a closed walk, then let $U=\bigcup_{t\in[L]}V(C_t)$, where each $C_t$ is the connected component of $F_t$ containing $x$. Then, because $T$ is a tree and each edge appears once in $\mathcal{T}$, a monotone $(x,x)$-walk $W$ in $T$ of maximum weight 
visits exactly the vertices of $U$ and the weight of $W$ is $w(U)$.
Since the connected components of each $F_t$ can be found in $\Oh(n)$ time, we can construct $U$ and compute its weight in $\Oh(nL)$ time. 

From now on we assume $y\neq x$. Consider the unique $(x,y)$-path $P=v_0,e_1,v_1,\dots, e_r,v_r$ in $T$ with $x=v_0$ and $y=v_r$. For any monotone $(x,y)$-walk $W$, $e_1,\ldots,e_r\in E(W)$ and $t(e_1)\leq\dots\leq t(e_r)$.  Moreover, 
$T$ has a monotone $(x,y)$-walk if and only if $t(e_1)\leq\dots\leq t(e_r)$. 
Suppose that $T$ has a monotone $(x,y)$-walk.

For $i\in[r-1]$, we say that $v_i$ is a \emph{concatenation point} if $t(e_{i})<t(e_{i+1})$. Because $T$ is a tree and each edge appears in $\mathcal{T}$ once, we have that 
$W$ is a monotone $(x,y)$-walk such that $V(W)$ is inclusion-maximal if and only if
$V(W)=\bigcup_{t\in[L]}V(C_t)$ and $E(W)=\bigcup_{t\in[L]}E(C_t)$, where each $C_t$ is a connected component of $F_t$ such that
\begin{enumerate}[(i)]
    \item\label{proof:first} either $t=t(e_i)$ for some $i\in[r]$ and  $C_t$ is a connected component of $F_t$ containing $e_i$,
\item or $t<t(e_1)$ and $C_t$ is the connected component of $F_t$ containing $v_0$,
 \item or $t(e_{i})< t<t(e_{i+1})$ for some $i\in[r-1]$ and $C_t$ is the connected component of $F_t$ containing $v_i$,
 \item\label{proof:last} or $t(e_r)<t$  and $C_t$ is the connected component of $F_t$ containing $v_r$.
\end{enumerate} 

Using this observation, we find the maximum weight of a $(x,y)$-walk. First, in $\Oh(n)$ time we find the unique $(x,y)$-path $P=v_0,e_1,v_1,\ldots, e_r,v_r$ in $T$  and check whether  $t(e_1)\leq\cdots\leq t(e_r)$. If the latter condition is not fulfilled, we report that a monotone $(x,y)$-walk  does not exist. Otherwise, we find the connected components $C_t$ of $F_t$ for $t\in[L]$ satisfying (\ref{proof:first})--(\ref{proof:last}). Then, we construct $U=\bigcup_{t\in[L]}^{L}V(C_t)$ and return $w(U)$. Because the connected components of each $F_t$ can be found in $\Oh(n)$ time, the overall running time is $\Oh(nL)$. 
\end{proof}
\end{toappendix}

To reduce the vertex weights in our kernelization algorithm, we use the approach proposed by Etscheid et al.~\cite{DBLP:journals/jcss/EtscheidKMR17} that is based on the result of Frank and Tardos~\cite{DBLP:journals/combinatorica/FrankT87}.

\begin{proposition}[\cite{DBLP:journals/combinatorica/FrankT87}]\label{prop:FT}
 There is an algorithm that, given a vector $\mathsf{w}\in \mathbb{Q}^r$ and an integer~$N$, in polynomial time finds a vector $\overline{\mathsf{w}}\in \mathbb{Z}^r$ with $\|\overline{\mathsf{w}} \|_{\infty}\leq 2^{4r^3}N^{r(r+2)}$ such that 
 $\mathsf{sign}(\mathsf{w}\cdot b)=\mathsf{sign}(\overline{\mathsf{w}}\cdot b)$ for all vectors $b\in\mathbb{Z}^r$ with $\|b\|_1\leq N-1$.
\end{proposition}

In fact, the algorithm of Frank and Tardos is strongly polynomial.

Now we are ready to prove the main result of the section. We remind that a pair $(P,Q)$ of subsets of $V(G)$ is called a \emph{separation} of $G$ if $P\cup Q=V(G)$ and there is no edge $xy$ with $x\in P\setminus Q$ and $y\in Q\setminus P$. The set $P\cap Q$ is a \emph{separator} and $|P\cap Q|$ is called the \emph{order} of a separation $(P,Q)$. 
If a separation has order one, then the single vertex of $P\cap Q$ is called a \emph{cut-vertex} and we say that this vertex is a separator. Notice that it may happen that $P\subseteq Q$ or $Q\subseteq P$ and we slightly abuse the standard notation, as it may happen that $P\cap Q$ does not separate $G$.
We also remind that an edge cut of $G$ is a partition $(P,Q)$ of $V(G)$ and the edge cut-set is the set of all edges $xy$ for $x\in P$ and $y\in Q$.

\begin{theorem}
	\label{thm:kernel-p}
	\WEX parameterized by $p = \mathfrak{m}-n+1$ admits a kernel of size $\Oh(p^4)$ for connected underlying graphs such that for the output instance $(\G=(G_1,\ldots,G_{L}),w,v,k)$, $\G$ has $\Oh(p)$ vertices and edges, and $L\in\Oh(p)$.
\end{theorem}

\begin{proof}
Let  $(\G=(G_1,\ldots,G_{L}),w,v,k)$ be an instance of \WEX and let $p = \mathfrak{m}-n+1$. Let also $G$ be the underlying graph of $\G$ and let $R$ be the set of red edges. We describe our kernelization algorithm as a series of reduction rules that are applied exhaustively whenever is is possible to apply  a rule. The rules modify $\GL$ and $w$. Whenever we say that a rule deletes a vertex $x$, this means that $x$ is deleted from $G_1,\ldots,G_L$ and $G$ together with the incident edges. Similarly, whenever we create a vertex, this vertex is added to every graph $G_i$ for $i\in[L]$ and $G$. When we either delete or add an edge, we specify $G_i$ where this operation is performed and also assume that the corresponding operation is done for $G$.

To describe the first rule, we need some auxiliary notation.
Let $(P,Q)$ be a separation of $G$ of order one with a cut-vertex $x$. We say that $(P,Q)$ is \emph{important} if 
\begin{enumerate}[(i)]
    \item\label{kernel-p:item:1} $G[P]$ is a tree with at least two vertices,
    \item\label{kernel-p:item:2} $v\in Q$ and $R\subseteq E(G[Q])$, and
    \item\label{kernel-p:item:3} $P$ is an inclusion-maximal set satisfying (\ref{kernel-p:item:1}) and (\ref{kernel-p:item:2}).
\end{enumerate}

Suppose that  $(P,Q)$ is an important separation of $G$ with a cut-vertex $x$. Let $a,b\in[L]$. Then, there is a separation $(P_1,P_2)$ of $H=G[P]$ with $x$ being the cut-vertex such that for any $y\in N_H(x)$, $y\in P_1$ if and only if $a\leq t(xy)\leq b$ and, moreover, $(P_1,P_2)$ is unique. We use $\T_{[a,b]}$ to denote the temporal graph $(G_1[P_1],\ldots,G_L[P_1])$ and $T_{[a,b]}=G[P_1]$. 
Observe that $\T_{[a,b]}$ may be a single-vertex temporal graph containing only $x$.
We write  $\T_{(a,b)}$ ($\T_{[a,b)}$ and $\T_{(a,b]}$, respectively) for $\T_{[a+1,b-1]}$ ($\T_{[a,b-1]}$ and $\T_{[a+1,b]}$, respectively) and we use the corresponding notation for the underlying trees. 
We also write $\T_a$ and $T_a$ instead of $\T_{[a,a]}$ and $T_{[a,a]}$.
Given an important separation $(P,Q)$ of $G$ with a cut-vertex $x$, we denote $A(Q)=\bigcup_{y\in N_{G[Q]}(x)}A(xy)$. We say that $i,j\in A(Q)$ are \emph{consecutive}  if $i<j$ and there is no $t\in A(Q)$ such $i< t<j$.

Observe that as cut-vertices and blocks of a graph can be found in linear time by standard algorithmic tools (see, e.g.,~\cite{DBLP:books/mg/CormenLRS01}), all important separations can be listed in linear time.

Given a cut vertex $x$ such that one side of the important separation $T$ is a tree, then
an optimal exploration can either enter and exist $T$ by $x$ or end
in some vertex of $T$. In both cases, using the algorithm from \autoref{lem:poly-trees} we
can compute the optimal exploration of $T$ and replace $T - x$ by $2$ vertices simulating
those two options giving us the first reduction rule.
\begin{redrule}\label{rule:trees}
    If there is an important separation $(P,Q)$ of $G$ with a cut-vertex $x$ such that either there is $i\in A(Q)$ such that $\T=\T_i$ has at least four vertices, or there are consecutive $i,j\in A(Q)$ such that $\T=\T_{(i,j)}$ has at least four vertices, or $\T=\T_{[1,i)}$ has at least four vertices for $i=\min A(Q)$, or $\T=\T_{(j,L]}$ has at least four vertices for $j=\max A(Q)$, then do the following for $\T$ and the underlying tree $T$ (see \autoref{fig:trees}):
\begin{itemize}
\item Call the algorithm from \autoref{lem:poly-trees} and compute the maximum weight $w_1$ of a monotone $(x,x)$-walk in $T$ and the maximum weight $w_2$ of a monotone $(x,y)$-walk, where maximum is taken over all $y\in V(T)$. Set  $t=\min_{u\in N_T(x)}t(xu)$.
 \item Delete the vertices of $\T$ in $\G$ except $x$, create a new vertex $y$, make it adjacent to $x$ in $G_t$  and set $w(y)=w_1-w(x)$.
 \item If $w_2>w_1$, then create a new vertex $z$, make it adjacent to $y$ in $G_L$ and set $w(z)=w_2-w_1$. 
\end{itemize}
\end{redrule}

\begin{figure}[ht]
\centering
\scalebox{0.7}{%
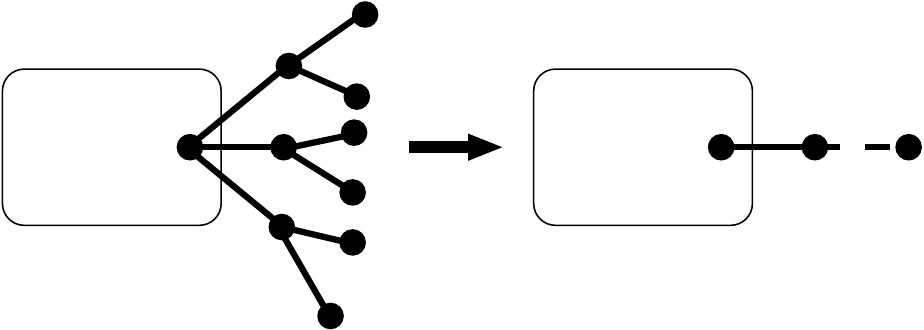}
\caption{The application of \autoref{rule:trees}.}%
\label{fig:trees}
\end{figure}

\begin{toappendix}
\begin{claim}\label{cl:trees}
    \autoref{rule:trees} is safe.
\end{claim}
\begin{claimproof}[Proof of \autoref{cl:trees}]
The rule is applied for four possible choices of $\T$. The safeness proofs for these cases are almost identical. Therefore, we provide it only for the case $\T=\T_{(i,j)}$ for consecutive $i,j\in A(Q)$ and then briefly discuss the differences in the proofs for other cases. 

Denote by $(\G',w',v,k)$ the instance of \WEX obtained by applying the rule and let $G'$ be the underlying graph of $\G'$. 
Observe that $T$ has an edge $xu$ with $t(xu)=t$. This implies that $w_1\geq w(x)+w(u)$ and, therefore, $w'(y)=w_1-w(x)>0$.  Because  $z$ is constructed only if $w_2>w_1$, we have that $w'(z)>0$. Thus, the constructed instance is feasible. Notice also that if $z$ is constructed, then $L\geq j>t$ and $t(xy)<t(yz)$ in~$\G'$. 

Suppose that $(\G,w,v,k)$ is a yes-instance of \WEX{}. By \autoref{obs:equiv}, this means that there is a monotone $(v,u)$-walk $W$ for some $u\in V(G)$ such that $w(W)\geq k$. If $W$ does not enter to any vertex of $T-x$, then $W$ is a monotone walk in $G'$ and $(\G',w',v,k)$ is a yes-instance. Assume that this is not the case and $W$ visits some vertices of $T-x$. We have that $i$ and $j$ are consecutive in $A(Q)$. Thus, $i<j$  and for every $h\in A(Q)$, either $h\leq i$ or $h\geq j$. 
Because $x$ is a cut-vertex and $i<t\leq t'<j$ for  $t=\min_{u\in N_T(x)}t(xu)$ and every $t'\in \{t(xu)\colon u\in N_T(x)\}$, 
the vertices of $V(W) \cap V(T)\setminus\{x\}$ are visited by a subwalk $W'$ of $W$ and $W'$ is a walk in $T$. Furthermore, if $u\notin V(T)\setminus \{x\}$, $W'$ is an $(x,x)$-walk, and if $u\in V(T)\setminus \{x\}$, then $W'$ is an $(x,u)$-walk.

Suppose that $W'$ is an $(x,x)$-walk. Then, $w(W)\leq w_1$ by the definition of $w_1$ and the total weight of the vertices of $T-x$ visited by $W'$ is at most $w'(y)\leq w_1-w(x)$. We construct the walk $\widehat{W}$ in $G'$ by replacing $W'$ by $x,xy,y,yx,x$. Then, $w'(\widehat{W})\geq w(W)\geq k$, that is, $(\G',w',v,k)$ is a yes-instance. 

Assume now that $W'$ is an $(x,u)$-walk for $u\in V(T)\setminus \{x\}$. If $w(W')\leq w_1$, then we replace $W'$ by $x,xy,y$ and have that for the obtained walk $\widehat{W}$, $w'(\widehat{W})\geq w(W)\geq k$.
Suppose that $w(W')> w_1$. Then, $w_2\geq w(W')> w_1$ and $z\in V(G')$. Since $w'(z)=w_2-w_1$, we can construct $\widehat{W}$ from $W$ by replacing  $W'$ by $x,xy,y,yz,z$ and obtain that $w'(\widehat{W})\geq w(W)\geq k$. This concludes the proof that  $(\G',w',v,k)$ is a yes-instance. 

\medskip
For the opposite direction, the arguments are very similar. Suppose that $(\G',w',v,k)$ is a yes-instance of \WEX and denote by $T'$ the tree created by the rule from $T$. Then, there is a monotone $(v,u)$-walk $\widehat{W}$ for some $u\in V(G')$ such that $w(\widehat{W})\geq k$. If $\widehat{W}$ does not enter $y$, then $\widehat{W}$ is a monotone walk in $G$ and $(\G,w,v,k)$ is a yes-instance. Suppose that $\widehat{W}$ enters $y$. Because $i<t<j$ and $x$ is a cut-vertex of $G'$, we can assume that $y$ is visited by a subwalk $\widehat{W}'$ of $\widehat{W}$ which is a walk in $T'$. 
If $u\notin V(T')\setminus \{x\}$, we can further assume that $\widehat{W}'$ is an $(x,x)$-walk and $\widehat{W}'=x,xy,y,yx,x$, because $t(yz)>t(xy)$ if $z$ exists. 
We construct the walk $W$ in $G$ by replacing $\widehat{W}'$ by a monotone $(x,x)$-walk $W'$ of weight $w_1$ in $T$ which exists by the choice of $w_1$. Notice that for each edge $e\in E(W')$, $i<t(e)<j$. Therefore, the replacement gives a monotone walk in $G$. We obtain that $w(W)=w'(\widehat{W})\geq k$. 
If $u\in V(T')\setminus\{x\}$, then we assume that  either $u=y$ and $\widehat{W}'=x,xy,y$ if $z$ does not exist or $u=z$ and $\widehat{W}'=x,xy,y,yz,z$ if $z\in V(T')$.  
In both cases, $w'(\widehat{W}')=w_2$. We construct $W$ by replacing $\widehat{W}'$ by a monotone $(x,u)$-walk $W'$ of weight $w_2$ in $T$ for some $u\in V(T)$ which exists by the definition of $w_2$. 
Then, we have that $W$ is a monotone walk in $G$ and  $w(W)=w'(\widehat{W})\geq k$. We conclude that $(\G,w,v,k)$ is a yes-instance. This concludes the proof for  the case $\T=\T_{(i,j)}$ for consecutive $i,j\in A(Q)$.

If $\T=\T_i$ for $i\in A(Q)$, then to modify the proof, we observe if the rule constructs $z$, then $i<L$. Hence, $t(xy)<t(yz)$ in $\G'$ if the rule constructs $z$. Observe also that for any $u\in N_T(x)$, $t(xu)=i$. This guarantees that we can apply the exchange arguments for the walks from above. In particular, we can assume without loss of generality that if a monotone walk $W$ enters $T-x$, then the vertices of $V(W) \cap V(T)\setminus \{x\}$ are visited by a subwalk of $W$.
For $\T=\T_{[1,i)}$, where $i=\min A(Q)$, we only observe that any monotone walk in $G$ can enter $T-x$ only if $x=v$, that is, $x$ is the source vertex. Finally, for 
 $\T=\T_{(j,L]}$, where $j=\max A(Q)$, we note that if $t=L$, then any monotone $(x,u)$-walk in $T$ can use only edges of $G_L$. Hence, the rule does not create $z$ in this case. Otherwise, if $z$ is created, then $t(xy)<t(yz)$ and we can apply the same arguments as for $\T=\T_{(i,j)}$. This concludes the proof.   
\end{claimproof}
\end{toappendix}

To state the next rules, we define edge cuts that are important for our algorithm. We say that an edge cut $(P,Q)$ of $G$ is \emph{important} if
\begin{enumerate}[(i)]
    \item\label{edge_cut:1} the edge-cut set is of size two and consists of edges with distinct endpoints, 
    \item\label{edge_cut:2} $G[P]$ is a tree, and
    \item\label{edge_cut:3} $v\in Q$ and $R\subseteq E(G[Q])$.
\end{enumerate}
Given an important edge cut $(P,Q)$ with an edge cut-set $\{x_1y_1,x_2y_2\}$, where $x_1,x_2\in P$, there is the unique $(x_1,x_2)$-path $S$ in the tree $G[P]$. We say that the path $y_1,y_1x_1,S,x_2y_2,y_2$ is the \emph{backbone} of the edge cut. We also denote $\T_{P,Q}=(G_1[(P\cup\{y_1,y_2\})]-y_1y_2,\ldots,G_L[P\cup\{y_1,y_2\}]-y_1y_2 )$.  Note that the underlying graph $T$ of $\T_{P,Q}$ is a tree and the unique $(y_1,y_2)$-path in $T$ is the backbone.  

Observe that all important edge cuts can be found in polynomial time.
For example, we can consider all pairs of edges without common endpoints and check whether the pair is the edge cut-set of an important edge cut.
We now construct three reduction rules that reduce the length of backbones.
Our next rule deletes irrelevant edges.

\begin{figure}[ht]
\centering
\scalebox{0.7}{
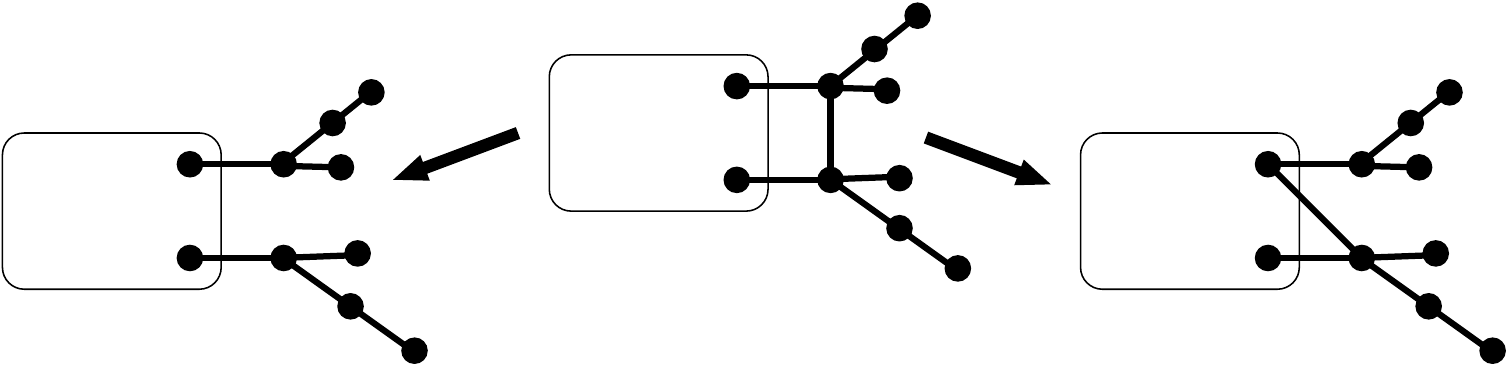}
\caption{The application of \autoref{rule:del} and \autoref{rule:cons}.}
\label{fig:del}
\end{figure}

\begin{redrule}\label{rule:del}
If there is an important edge cut $(P,Q)$ of $G$ with a backbone $v_0,e_1,v_1,e_2,v_2,e_3,v_3$ such that $t(e_1)>t(e_2)$ and $t(e_3)>t(e_2)$,  then modify $\G$ and $G$ by deleting $e_2$ from $G_{t(e_2)}$ (see \autoref{fig:del}). 
 Furthermore, if the obtained underlying graph $G$ is disconnected, then delete the vertices of the (unique) connected component that does not contain the source vertex $v$.
\end{redrule}

The safeness of the rule immediately follows from the observation that no monotone $(v,u)$-walk can use $e_2$, because any monotone  walk can enter $P$ only via $e_1$ or $e_3$.

\medskip
The next rule aims at consecutive edges in backbones that occur in the same $G_t$.


\begin{redrule}\label{rule:cons}
If there is an  an important edge cut $(P,Q)$ of $G$ with a backbone $v_0,e_1,v_1,e_2,v_2,e_3,v_3$ such that $t(e_1)=t(e_2)$, then modify $\G$ and $G$ by deleting $e_1$ from $G_{t(e_1)}$ and adding $v_0v_2$ in $G_{t(e_1)}$ (see \autoref{fig:del}).
\end{redrule}

\begin{toappendix}
\begin{claim}\label{cl:cons}
\autoref{rule:cons} is safe.
\end{claim}
\begin{claimproof}[Proof of \autoref{cl:cons}]
Denote by $(\G',w',v,k)$ the instance of \WEX obtained by applying the rule and let $G'$ be the underlying graph of $\G'$. Suppose that  that $(\G,w,v,k)$ is a yes-instance of \WEX. By \autoref{obs:equiv}, this means that there is a monotone $(v,u)$-walk $W$ for some $u\in V(G)$ such that $w(W)\geq k$. We construct the walk $W'$ by replacing each occurrence  of the edge $e_1$ by  $v_0v_2,v_2,e_2$ or $e_2,v_2,v_2v_0$, respectively, depending on the direction on which $e_1$ is crossed by the walk. 
This way we obtain a walk in $G'$ and, because $t(e_1)=t(e_2)$, this is a monotone $(v,u)$-walk. By the construction,  $w'(W')\geq w(W)\geq k$ and we have that $(\G',w',v,k)$ is a yes-instance. For the opposite direction, the arguments are almost the same. The difference is that we replace each occurrence of $v_0v_2$ in a walk in $G'$ by either $e_1,v_1,e_2$ or, symmetrically, by $e_2,v_1,e_1$. This concludes the proof.
\end{claimproof}
\end{toappendix}

\begin{figure}[ht]
\centering
\scalebox{0.7}{%
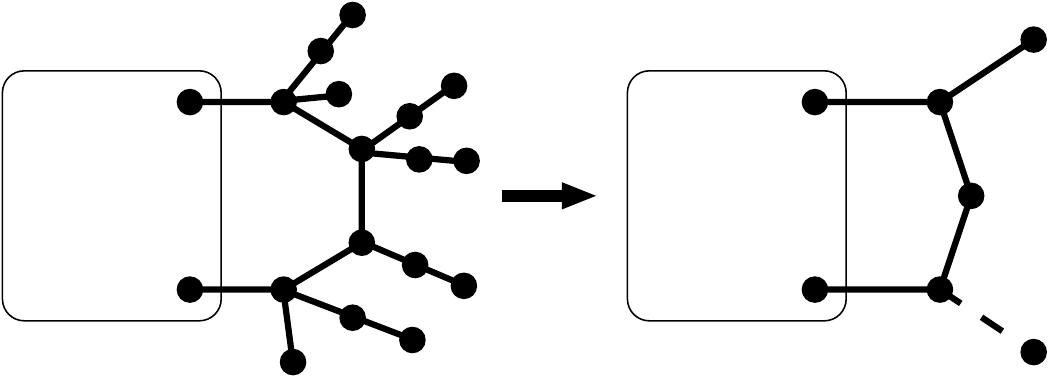}
\caption{The application of \autoref{rule:contr}.}%
\label{fig:contr}
\end{figure}

It remains to shorten monotone backbones.
\begin{redrule}\label{rule:contr}
If there is an important edge cut $(P,Q)$ of $G$ with a backbone $v_0,e_1,\ldots,e_5,v_5$ such that $t(e_1)<\cdots<t(e_5)$, then modify $\G$ and $G$ by doing the following for $\T=\T_{P,Q}$ with its underlying tree $T$ (see \autoref{fig:contr}):
\begin{itemize}
\item Call the algorithm from \autoref{lem:poly-trees} and compute the maximum weight $w_1$ of a monotone $(v_0,v_0)$-walk in $T$,  the maximum weight $w_2$ of a monotone $(v_0,v_5)$-walk and the maximum weight $w_3$ of a monotone $(v_5,v_5)$-walk.
\item Delete the vertices of $V(T)\setminus \{v_0,v_5\}$, create new vertices $u_1,u_2,u_3$, make $u_1$ adjacent to $v_0$ in $G_{t(e_1)}$, make $u_2$ adjacent to $u_1$ and $u_3$ in $G_{t(e_2)}$ and $G_{t(e_4)}$, respectively, and make $u_3$ adjacent to $v_5$ in $G_{t(e_5)}$.
\item  Set $w(u_1)=w_1-w(v_0)$, $w(u_3)=w_3-w(v_5)$ and $w(u_2)=w_2-w_1-w_3$.    
\item Call the algorithm from \autoref{lem:poly-trees} and compute the maximum weight $w_0$ of a monotone $(v_0,u)$-walk in $T-v_5$, where the maximum is taken over all $u\in V(T)\setminus\{v_5\}$. 
Create a vertex $y_1$, make $y_1$ adjacent to $u_1$ in $G_{t(e_3)}$ and set $w(y_1)=w_0-w(v_0)-w(u_1)-w(u_2)$.
\item Call the algorithm from \autoref{lem:poly-trees} and compute the maximum weight $w_5$ of a monotone $(v_5,u)$-walk in $T$, where the maximum is taken over all $u\in V(T)$. If $w_5>w_3$, then create a vertex $y_3$, make $y_3$ adjacent to $u_3$ in $G_{L}$ and set $w(y_3)=w_5-w(v_5)-w(u_3)$.
\end{itemize}
\end{redrule}

\begin{toappendix}
\begin{claim}\label{cl:contr}
\autoref{rule:contr} is safe.
\end{claim}
\begin{claimproof}[Proof of \autoref{cl:contr}]
Denote by $(\G',w',v,k)$ the instance of \WEX obtained by applying the rule and let $G'$ be the underlying graph of $\G'$. 

First, we argue that the instance  $(\G',w',v,k)$ is feasible in the sense that the weights are positive. Notice that by the construction, $w_1\geq w(v_0)+w(v_1)$. Therefore, $w'(u_1)>0$. Similarly, it holds that $w'(u_3)>0$. Notice that given a monotone $(v_0,v_0)$-walk $W_1$, a monotone $(v_0,v_5)$-walk $W_2$ and a monotone $(v_5,v_5)$-walk $W_3$ in $T$, we can concatenate them and obtain the $(v_0,v_5)$-walk $W$.
The walk $W$ is monotone, $V(W_1)\subseteq V(W)$ and $V(W_3)\subseteq V(W)$. Because $t(e_1)<t(e_2)<t(e_4)<t(e_5)$, it holds that $v_2,v_3\in V(W_2)\setminus (V(W_1)\cup V(W_3))$. This implies that $w_2>w_1+w_3$ and $w(u_2)$ is positive. Notice that  $w_0\geq w_2-w(v_5)\geq w(v_0)+w(u_1)+w(u_2)+w(u_3)$. Hence, 
 $w(y_1)=w_0-w(v_0)-w(u_1)-w(u_2)\geq w(u_3)>0$. 
Finally, we create $y_3$ only if $w_5>w_3$. Then, 
$w(y_3)=w_5-w(v_5)-w(u_3)=w_5-w(v_5)-(w_3-w(v_5))=w_5-w_3>0$. Observe also that if $t(e_5)=L$, then $w_5=w_3$, because any monotone $(v_5,u)$-walk can use only edges of $G_L$. This means that if $y_3$ is constructed by the rule, then $t(e_5)<L$.

Suppose that $(\G,w,v,k)$ is a yes-instance of \WEX. By \autoref{obs:equiv}, this means that there is a monotone $(v,u)$-walk $W$ for some $u\in V(G)$ such that $w(W)\geq k$. 
If $W$ does not use any edge of $T$, then $W$ is a monotone walk in $G'$ and $(\G',w',v,k)$ is a yes-instance. Assume that this is not the case and $W$ visits some vertices of $T-\{v_0,v_5\}$. We consider three cases depending on the behavior of $W$.

\subparagraph{Case~1.} There is an edge of the backbone that is not used by $W$. 

If $e_5\notin E(W)$, then, because $v\in Q$ and  $W$ can enter $T-\{v_0,v_5\}$ only through $e_1$, we can assume without loss of generality that the vertices of $V(W) \cap V(T) \setminus \{v_0,v_5\}$ are visited by a subwalk $W'$ of $W$ that is a monotone $(v_0,z)$-walk in $T$ for some $z\in V(T)$. If $z=v_0$, then we construct the walk $\widehat{W}$ in $G'$ by replacing $W'$ by $v_0,v_0u_1,u_1,v_0u_1,v_0$. Then, because $w(W')\leq w_1$ by definition of $w_1$ and $w'(u_1)=w_1-w'(v_0)$, we have that $w'(\widehat{W})\geq w(W)\geq k$. If $z\neq v_0$, then $z=u$. We construct $\widehat{W}$ by replacing $W'$ by $v_0,v_0u_1,u_1, u_1u_2,u_2,u_1u_2,u_1,u_1y_1,y_1$.
Because $e_5$ is not used by $W'$,
$w(W')\leq w_0$ by definition of $w_0$. Since $w(y_1)=w_0-w(v_0)-w(u_1)-w(u_2)$,  we obtain that $w'(\widehat{W})\geq w(W)\geq k$. In both cases,  we conclude that $(\G',w',v,k)$ is a yes-instance.

Suppose that $e_1\notin E(W)$. Now we can assume that the vertices of $V(W) \cap V(T)\setminus\{v_0,v_5\}$ are visited by a subwalk $W'$ of $W$ that is a monotone $(v_5,z)$-walk in $T$ for some $z\in V(T)$. If $z=v_5$, then we construct the walk $\widehat{W}$ in $G'$ by replacing $W'$ by $v_5,v_5u_3,u_3,v_5u_3,v_5$. Then, because $w(W')\leq w_3$ and $w'(u_3)=w_3-w'(v_5)$, we have that $w'(\widehat{W})\geq w(W)\geq k$. If $z\neq v_5$, then $z=u$. We construct $\widehat{W}$ by replacing $W'$ by $v_5,v_5u_3,u_3$ if $w_5=w_3$ and by $v_5,v_5u_3,u_3,u_3y_y,y_3$ if $w_5>w_3$. 
Because $e_1$ is not used by $W'$,
$w(W')\leq w_5$ and  we obtain that $w'(\widehat{W})\geq w(W)\geq k$. In both cases,  we have that  $(\G',w',v,k)$ is a yes-instance.

Assume now that $e_1,e_5\in E(W)$. Then, we can assume that the vertices of $V(W) \cap V(T)\setminus\{v_0,v_5\}$ are visited by exactly two disjoint subwalks $W_1$ and $W_2$ of $W$, where $W_1$ is a $(v_0,v_0)$-subwalk of $W$ and $W_2$ is a $(v_5,z)$-subwalk of $W$ for some $z\in V(T)$. Then, we replace $W_1$ in exactly the same way as for $e_5\notin E(W)$ and $z=v_0$. Similarly, $W_2$ is replaces in exactly the same way as for $e_1\notin E(W)$. The subwalks used for the replacements are disjoint and we conclude that $(\G',w',v,k)$ is a yes-instance.

\subparagraph{Case~2.} The edges of  the backbone are in $W$ and $u\notin V(T)\setminus\{v_0,v_5\}$.  
Since $t(e_1)<\dots<t(e_5)$, we can assume that  the vertices of $V(W) \cap V(T)\setminus\{v_0,v_5\}$ are visited by a subwalk $W'$ of $W$ that is a monotone $(v_0,v_5)$-walk in $T$. Then, we construct 
the walk $\widehat{W}$ in $G'$ by replacing $W'$ by $v_0,v_0u_1,u_1,u_1u_2,u_2,u_2u_3,u_3,u_3v_5,v_5$. Then, because $w(W')\leq w_2$ and by the definitions of the weights of $u_1,u_2,u_3$,  we have that $w'(\widehat{W})\geq w(W)\geq k$ implying that  $(\G',w',v,k)$ is a yes-instance.

\subparagraph{Case~3.} The edges of  the backbone are in $W$ and $u\in V(T)\setminus\{v_0,v_5\}$.  
Because all the edges of the backbone are in $W$ and $t(e_1)<\cdots<t(e_5)$, we can assume that the vertices of $V(W) \cap V(T)\setminus\{v_0,v_5\}$ are visited by two subwalks $W_1$ and $W_2$ of $W$ such that $W_1$ is a monotone $(v_0,v_5)$-walk in $T$ and $W_2$ is a monotone $(v_5,u)$-walk.  Let $W_3$ be a monotone $(v_5,v_5)$-walk in $T$ of maximum weight $w_3$. 
We can assume without loss of generality that $W_3$ is a subwalk of both $W_1$ and $W_2$; otherwise, we can include $W_3$ in these walks in the end and the beginning, respectively. 
Notice that for every $e\in E(W_3)$, $t(e)=t(e_5)$. Then, because $t(e_1)<\cdots<t(e_5)$ and $W_1,W_2$ are monotone, we have that 
$E(W_1)\cap E(W_2)=E(W_3)$. We construct 
the walk $\widehat{W}$ in $G'$ by replacing $W_1$ by $v_0,v_0u_1,u_1,u_1u_2,u_2,u_2u_3,u_3,u_3v_5,v_5$. If $w_5=w_3$, we replace $W_2$ by $v_5,u_3v_5,u_3$. Otherwise, if $w_5>w_3$, $y_3\in V(G')$ and we replace $W_2$ by $v_5,u_3v_5,u_3,u_3y_3,y_3$. Because $E(W_1)\cap E(W_2)=E(W_3)$, $w(V(W_1)\cup V(W_2))=w(W_1)+w(W_2)-w_3$. By the definitions of the weights of $u_1,u_2,u_3$ and $y_3$, we obtain that $w'(\widehat{W})\geq w(W)\geq k$. Hence, $(\G',w',v,k)$ is a yes-instance. This concludes the proof for the forward direction.

\medskip
For the opposite direction, the proof uses similar arguments as we reverse the exchange arguments. Denote by $T'$ the tree constructed from $T$ by the rule.
Let $W_1$ be  a monotone $(v_0,v_0)$-walk of weight $w_1$ in $T$, $W_2$ be a monotone $(v_0,v_5)$-walk of weight $w_2$ and $W_3$ be a monotone $(v_5,v_5)$-walk of weight $w_3$. Let also $W_1'$ be a monotone $(v_0,z)$-walk in $T-v_5$ of weight $w_0$ and let $W_3'$ be a monotone $(v_5,z)$-walk in $T$ of weight $w_5$ for some $z\in V(T)$. Observe that because $t(v_0u_1)<t(u_1u_2)<t(u_2u_3)<t(u_3v_5)$, $V(W_1)\cap V(W_3)=\emptyset$ and $V(W_1)\cap V(W_3')=\emptyset$. Note also that  $V(W_1),V(W_3)\subseteq V(W_2)$, because otherwise we can concatenate $W_1$ and $W_3$ with $W_2$ and increase the weight of $W_2$. 

Suppose that   $(\G',w',v,k)$ is a yes-instance of \WEX{}. Then, there is a monotone $(v,u)$-walk $\widehat{W}$ in $G'$ with $w'(\widehat{W})\geq k$. If $\widehat{W}$ does not enter $T'-\{v_0,v_5\}$, then $\widehat{W}$ is a monotone walk in $G'$ and $(\G,w,v,k)$ is a yes-instance. Assume that $\widehat{W}$ visits some vertices of $T'-\{v_0,v_5\}$. We consider three cases that are analogical to Cases~1--3 above.

\subparagraph{Case~1.} There is an edge of the backbone of $T'$ that is not used by $W$. 

If $u_3v_5\notin E(\widehat{W})$, then we can assume without loss of generality that the vertices of $V(\widehat{W}) \cap V(T')\setminus\{v_0,v_5\}$ are visited by a subwalk $\widehat{W}'$ of $\widehat{W}$ that is a monotone $(v_0,z)$-walk in $T'$ for some $z\in V(T')$. If $z=v_0$, then we construct the walk $W$ in $G$ by replacing $\widehat{W}'=v_0,v_0u_1,u_1,v_0u_1,v_0$ by $W_1$. If $z\neq v_0$, then $z=u$ and we replace $\widehat{W}'$ by $W_1'$. In both cases, we obtain that
$w(W)\geq w'(\widehat{W})\geq  k$. Thus, $(\G,w,v,k)$ is a yes-instance. 

If $v_0u_1\notin E(\widehat{W})$, then we can assume that the vertices of $V(\widehat{W}) \cap V(T')\setminus\{v_0,v_5\}$ are visited by a subwalk $\widehat{W}'$ of $\widehat{W}$ that is a monotone $(v_5,z)$-walk in $T$ for some $z\in V(T')$. If $z=v_5$, then  the walk $W$ in $G$ is constructed by replacing $\widehat{W}'=v_5,v_5u_3,u_3,v_5u_3,v_5$ by $W_3$. If $z\neq v_5$, then $z=u$ and we replace $\widehat{W}'$ by $W_3'$. We have that
$w(W)\geq w'(\widehat{W})\geq  k$ and  $(\G,w,v,k)$ is a yes-instance. 

Assume that $v_0u_1,u_3v_5\in E(\widehat{W})$. Then, we can assume that the vertices of $V(\widehat{W}) \cap V(T')\setminus\{v_0,v_5\}$ are visited by exactly two disjoint subwalks $\widehat{W}_1$ and $\widehat{W}_2$ of $\widehat{W}$, where $\widehat{W}_1$ is a $(v_0,v_0)$-subwalk of $\widehat{W}$ and $\widehat{W}_2$ is a $(v_5,z)$-subwalk for some $z\in V(T')$. Then, we replace $\widehat{W}_1$ by $W_1$ in exactly the same way as for $u_3v_5\notin E(\widehat{W})$ and $z=v_0$. Similarly, $\widehat{W}_2$ is replaces by either $W_3$ or $W_3'$ in exactly the same way as for $v_0u_1\notin E(\widehat{W})$. The subwalks used for the replacements are disjoint and we conclude that $(\G,w,v,k)$ is a yes-instance.

\subparagraph{Case~2.} The edges of  the backbone of $T'$ are in $\widehat{W}$ and $u\notin V(T')\setminus\{v_0,v_5\}$.
In this case, we can assume that the vertices of $V(\widehat{W}) \cap V(T')\setminus\{v_0,v_5\}$ are visited by the subwalk $\widehat{W}'=v_0,v_0u_1,u_1,u_1u_2,u_2,u_2u_3,u_3,u_3v_5,v_5$ and 
we construct $W$ by replacing $\widehat{W}'$ by $W_2$. Then, $w(W)\geq w'(\widehat{W})\geq  k$ and  $(\G,w,v,k)$ is a yes-instance.

\subparagraph{Case~3.} The edges of  the backbone of $T'$ are in $W$ and $u\in V(T')\setminus\{v_0,v_5\}$. Now we can assume that the vertices of $V(\widehat{W}) \cap V(T')\setminus\{v_0,v_5\}$ are visited by two subwalks 
$\widehat{W}_1=v_0,v_0u_1,u_1,u_1u_2,u_2,u_2u_3,u_3,u_3v_5,v_5$ and $\widehat{W}_2$ of $\widehat{W}$ such that either $y_3\notin V(G')$ and $\widehat{W}_2=v_5,u_3v_5,u_3$ or $y_3\in V(G')$ and $\widehat{W}_2=v_5,u_3v_5,u_3,u_3y_3,y_3$. We replace $\widehat{W}_1$ by $W_2$ and $\widehat{W}_2$ by $W'_3$. Again, $w(W)\geq w'(\widehat{W})\geq  k$ and  $(\G,w,v,k)$ is a yes-instance. This concludes the proof for the backward direction.
\end{claimproof}
\end{toappendix}

\noindent We observe that by the definitions of Rules~\ref{rule:trees}--\ref{rule:contr}, we immediately obtain the following claim.

\begin{claim}\label{cl:p}
Rules~\ref{rule:trees}--\ref{rule:contr} do not increase the parameter   $p = \mathfrak{m} - n+1$. 
\end{claim}

We show that in the instance of \WEX obtained by the exhaustive application of Rules~\ref{rule:trees}--\ref{rule:contr}, the graph has bounded size.

\begin{claim}\label{cl:size}
    If Rules~\ref{rule:trees}--\ref{rule:contr} are not applicable for $(\G,w,v,k)$, then $|V(\mathcal{G})|\leq 324p $ and $|E(\mathcal{G})|\leq 326p$.
\end{claim}

\begin{toappendix}
\begin{claimproof}[Proof of \autoref{cl:size}]
Suppose that neither of Rules~\ref{rule:trees}--\ref{rule:contr} can be applied for $(\G,w,v,k)$. We use \autoref{lem:struct}. By this lemma, there is a set of vertices $X\subseteq V(G)$ of size at most $4p$ such that (i)~$v\in X$, (ii)~every red edge is in $G[X]$, (iii) $G-E(G[X])$ is a forest and (iv)~each connected component $F$ of $G-X$ is a tree with the properties that each vertex $x\in X$ has at most one neighbor in $F$ and 
\begin{enumerate}[(a)]
\item either $X$ has a unique vertex $x$ that has a neighbor in $F$,
\item or $X$ contain exactly two vertices $x$ and $y$ having neighbors in $F$. 
\end{enumerate}
We assume that $X$ is an inclusion-minimal set with these properties. We say that a connected component $F$ of $G-X$ satisfying (\ref{lem:struct:a}) or (\ref{lem:struct:b}) is an \emph{(a)-type} or \emph{(b)-type} component, respectively. 
Also by the lemma,  if $q$ is the number of (b)-type components, then $q\leq 4p-1$ and $G[X]$ has at most $5p-q-1$ edges.

Let $x\in X$ and consider the (a)-type components $F_1,\ldots,F_r$ that have $x$ as a neighbors for some $r\geq 1$. Let $P_x=\{x\}\cup \bigcup_{i\in [r]}V(F_i)$ and $Q_x=(V(G)\setminus P)\cup\{x\}$. Because $X$ is an inclusion-minimal set satisfying  (i)--(iv), $(P_x,Q_x)$ is an important separation and $x$ is its cut-vertex. Because \autoref{rule:trees} cannot be applied, we have that 
$|P_x\setminus \{x\}|\leq 4|A(Q_x)|+2$ and $|E(G[P_x])|\leq 4|A(Q_x)|+2$.
In $Q_x$, $x$ is either adjacent to a vertex in $X$ or to a (b)-type components. Taking into account that the edges connecting $X$ with the (b)-type components appear exactly in one time step, we have that 
\begin{equation*}
\sum_{x\in X}|A(Q_x)|\leq 2\sum_{e\in E(G[X])}m(e)+2q\leq 2p+2|E(G[X])|+2q\leq 12p-2.
\end{equation*}
Thus, the total number of vertices in the (a)-type components $F$ of $G-X$ is at most $56p-8$ and these vertices are incident to at most $56p - 8$ edges. 

Consider a (b)-type component $F$ of $G-X$. Let $x$ and $y$ be the neighbors of $F$ in $X$ and denote by $x'$ and $y'$ the neighbors of $x$ and $y$ in $F$. Let $S$ 
be the $(x',y')$-path in the tree~$F$. We claim that $S$ has at most 6 vertices.
\todo[inline]{The worse case seems to be two increasing path to the same vertex. Maybe by spitting this vertex, we could reduce this size}

The proof is by contradiction. Assume that $S$ has at least 7 vertices. Observe that $(P,Q)$, where $P=V(F)$ and $Q=V(G)\setminus P$, is an important  edge cut. Let $v_0,e_1,\ldots,e_\ell,v_\ell$ be its backbone. Note that 
$\ell\geq 8$. Because \autoref{rule:cons} in not applicable, $t(e_i)\neq t(e_{i+1})$ for every $i\in[\ell-1]$. Since \autoref{rule:del} cannot be applied, for every $i\in[\ell-2]$,
ether $t(e_i)<t(e_{i+1})<t(e_{i+2})$, or $t(e_i)>t(e_{i+1})>t(e_{i+2})$, or $t(e_i)<t(e_{i+1})$ and $t(e_{i+1})>t(e_{i+2})$. Observe the last possibility $t(e_i)<t(e_{i+1})$ and  $t(e_{i+1})>t(e_{i+2})$ can occur only for a unique $i\in[\ell-2]$. Because $\ell\geq 8$, ether $t(e_1)<\dots<t(e_5)$ or $t(e_{\ell-4})>\dots >t(e_\ell)$. We can assume by symmetry that  $t(e_1)<\dots<t(e_5)$. Note that $\{e_1,e_5\}$ is the cut-set of an important edge cut 
$(P',Q')$, where $P'$ is the set of vertices of the connected component of $F-e_5$ containing $v_1=x'$ and $Q'=V(G)\setminus P'$. Then, $v_0,e_1,\dots,e_5,v_5$ is the corresponding backbone.
However, we are able to apply \autoref{rule:contr} for $(P',Q')$. This contradicts our assumption that the reduction rules are not applicable and proves that $S$ has at most 6 vertices.

Let $z\in V(S)$ and consider a connected component $H$ of $F-E(S)$ containing $z$. If $H$ has at least one edge, then  $(P,Q)$, where $P=V(H) $ and $Q=(V(G)\setminus V(H))\cup\{z\}$ is an important separation.  
Recall that each edge of $S$ and the edges $xx'$ and $yy'$ appear only once in the temporal graph. Thus, $|A(Q)|\leq 2$. 
Because \autoref{rule:trees} cannot be applied, we obtain that $|V(H)|\leq 4|A(Q)|+3\leq 11$ and $|E(H)|\leq 10$. Because $S$ has at most 6 vertices, we conclude that $F$ has at most $66$ vertices and the vertices of $F$ are adjacent to at most $67$ edges of $G$. 

Because the total number of (b)-type components is  $q\leq 4p-1$, the total number of vertices in these components is at most $264p-66$ and the vertices of these components are incident to at most
$67q$
edges of $G$.

Summarizing and using that $q\leq 4p-1$, $|X| \leq 4p$ and $|E(G)| \leq 5p - q - 1$, we obtain that 
\begin{equation*}
    |V(\mathcal{G}| = |V(G)|\leq |X|+(56p-8)+(264p-66)\leq 324p
\end{equation*}
and
\begin{equation*}
    |E(\mathcal{G})| \leq |E(G)| + p\leq |E(G[X])|+(56p-8)+67q + p\leq 326p.
\end{equation*}
This concludes the proof.
 \end{claimproof}
\end{toappendix}

 The bound on the number of edges of $G$ allows to reduce the value of $L$. 

 \begin{redrule}\label{rule:level}
     If there is $t\in[L]$ such that $G_t$ has no edge, then delete $G_t$ from $\mathcal{G}$.
 \end{redrule}

 The rule is trivially safe, because edgeless graphs are irrelevant for the exploration of $\G$. 
 After applying the rule $L\leq 326p$ as $|E(\mathcal{G})|\leq 326p$.

 \medskip
Our last aim is to reduce the weights. For this, we use the algorithm from \autoref{prop:FT}. Let $n=|V(G)|$ and let $V(G)=\{v_1,\ldots,v_n\}$. We define $r=n+1$ and consider the vector $\mathsf{w}={(w_0,w_1,\ldots,w_n)}^\intercal\in\mathbb{Z}^r$, where $w_0=k$ and $w_i=w(v_i)$ for $i\in[n]$.

\begin{redrule}\label{rule:weights}
Apply the algorithm from \autoref{prop:FT} for $\mathsf{w}$ and $N=r+1$ and find the vector $\overline{\mathsf{w}}=(\overline{w}_0,\ldots,\overline{w}_n)$. Set
$k:=\overline{w}_0$ and set $w(v_i):=\overline{w}_i$ for $i\in[n]$.
\end{redrule}

To see that the rule is safe, let $k'=\overline{\mathsf{w}}_0$ and let $w'(v_i)=\overline{w}_i$ for $i\in[n]$.
Note that by the choice of $N$, for each vector $b\in{\{-1,0,1\}}^r$, we have that 
$\mathsf{sign}(\mathsf{w}\cdot b)=\mathsf{sign}(\overline{\mathsf{w}}\cdot b)$. This implies that the new weights $w'(x)$ and $k'$ are positive integers and for every $U\subseteq V(G)$, $w(U)\geq k$ if and only if $w'(U)\geq k'$.

\medskip
By \autoref{prop:FT}, we obtain that $k\leq 2^{4n^3}{(n+2)}^{(n+1)(n+2)}$ and the same upper bound holds for $w(x)$ for every $x\in V(G)$. Because $|V(G)|=\Oh(p)$, we have that we need $\Oh(p^3)$ bits to encode $k$ and the weight of each vertex. Then, the total bit-length of the encoding of the weights and $k$ is $\Oh(p^4)$.  Taking into account that $|V(G)|=\Oh(p)$, $|E(G)|=\Oh(p)$ and $L=\Oh(p)$, we conclude that we obtained a kernel of size $\Oh(p^4)$. 

Because important separations and important edge cuts can be found in polynomial time and the algorithm from \autoref{lem:poly-trees} is polynomial, we have that Rules~\ref{rule:trees}--\ref{rule:contr} can be exhaustively applied in polynomial time. It is trivial that \autoref{rule:level} can be applied in polynomial time. Because the algorithm from \autoref{prop:FT} is polynomial, \autoref{rule:weights} requires polynomial time. Therefore, the overall running time of our kernelization algorithm is polynomial. This concludes the proof.
\end{proof}

In~\cite{ErlebachS22Journal}, Erlebach and Spooner proved that \textsc{NS-TEXP} can be solved in $\Oh(2^nLn^3)$ time. 
This fact together with \autoref{thm:kernel-p} implies the following corollary. 

\begin{corollary}\label{cor:fpt}
    \WEX{} parameterized by $p = \mathfrak{m}-n+1$ can be solved in $2^{\Oh(p)}{(nL)}^{\Oh(1)}$ time on temporal graphs with connected underlying graphs.
\end{corollary}

\begin{proof}
Given an instance $(\GL,w,v,k)$ of \WEX, we first apply the kernelization algorithm from \autoref{thm:kernel-p}. Then, we guess the set of vertices $U$ that 
is visited by an exploration walk. Because the number of vertices is $\Oh(p)$, we have $2^{\Oh(p)}$ choices of $U$. For each $U$, we verify whether $w(U)\geq k$ and whether $(\G[U]=(G_1[U],\ldots,G_L[U]),v)$ is a yes-instance of \textsc{NS-TEXP} using the algorithm of  Erlebach and Spooner~\cite{ErlebachS22Journal}. The overall running time of the algorithm is $2^{\Oh(p)}{(nL)}^{\Oh(1)}$. 
\end{proof}
\begin{toappendix}
We remark that using simplified versions of the reduction rules from the kernelization algorithm for \WEX, we can obtain a kernel of linear size for \textsc{NS-TEXP} for the parameterization by $p = \mathfrak{m}-n+1$.

We conclude this section by observing that \WEX{} admits an easy kernel for the parameterization by the number of non-appearances of edges. 

\begin{observation}\label{obs:non-app}	
    \WEX{} parameterized by $q=|E(G)|L-\mathfrak{m}$ admits a kernel of size $\Oh(q^4)$ for connected underlying graphs such that for the output instance $(\G=(G_1,\ldots,G_{L}),w,v,k)$, $\G$ has at most $q+1$ vertices,
at most $(q-1)L$ edges, and $L\leq q$.
\end{observation}

\begin{proof}
    Let  $(\G=(G_1,\ldots,G_{L}),w,v,k)$ be an instance of \WEX{} with $G$ being the underlying graph of $\G$. Let also
$q =|E(G)|L-\mathfrak{m}=\sum_{t=1}^L|E(G)\setminus E(G_t)|$. The kernelization algorithms boils down to the following straightforward reduction rules.

 \begin{redrule}\label{rule:edge-contr}
 If there is $xy\in E(G)$ such that $xy\in E(G_t)$ for every $t\in [L]$, then contract the edge $xy$ in every $G_t$ and set the weight of the obtained vertex $z_{xy}$ to be  $w(z_{xy}):=w(x)+w(y)$. Furthermore, if $v\in\{x,y\}$, then set the source vertex $v:=z_{xy}$.
  \end{redrule}
  
 \begin{redrule}\label{rule:single}
 If there is $t\in [L]$ such that $G_t$ is connected, then
 \begin{itemize}
 \item if $w(V(G))\geq k$, then return a trivial yes-instance of constant size;
 \item otherwise, if $w(V(G))<k$, then return a trivial no-instance.
\end{itemize}
 \end{redrule}
 
 Notice that after exhaustive applying Reduction Rule~\ref{rule:edge-contr}, $G$ has at most $q$ edges. Since $G$ is connected, $G$ has at most $q+1$
 vertices. Because of Reduction Rule~\ref{rule:single}, each $G_t$ has at most $q-1$ edges since at least one edge of $G$ does not appear in $G_t$. Thus, the total
 number of edges is at most $(q-1)L$.  Reduction Rule~\ref{rule:single} also guarantees that $L\leq q$, because at least one  edge of $G$ is missing in each $G_t$. 
 To complete the proof, observe that the weights can be reduced in exactly the same way as in the proof of Theorem~\ref{thm:kernel-p} using Proposition~\ref{prop:FT} (see Reduction Rule~\ref{rule:weights}).
 This concludes the proof.
\end{proof}
\end{toappendix}

%% file: Struct.pdf_t
\begin{picture}(0,0)%
\includegraphics{Struct.pdf}%
\end{picture}%
\setlength{\unitlength}{3947sp}%
\begingroup\makeatletter\ifx\SetFigFont\undefined%
\gdef\SetFigFont#1#2#3#4#5{%
  \reset@font\fontsize{#1}{#2pt}%
  \fontfamily{#3}\fontseries{#4}\fontshape{#5}%
  \selectfont}%
\fi\endgroup%
\begin{picture}(3520,2654)(543,-1463)
\put(1153,-596){\makebox(0,0)[lb]{\smash{{\SetFigFont{12}{14.4}{\rmdefault}{\mddefault}{\updefault}$X$}}}}
\put(793,-283){\makebox(0,0)[lb]{\smash{{\SetFigFont{12}{14.4}{\rmdefault}{\mddefault}{\updefault}$v$}}}}
\end{picture}%

%% file: Rule-one.pdf_t
\begin{picture}(0,0)%
\includegraphics{Rule-one.pdf}%
\end{picture}%
\setlength{\unitlength}{3947sp}%
\begingroup\makeatletter\ifx\SetFigFont\undefined%
\gdef\SetFigFont#1#2#3#4#5{%
  \reset@font\fontsize{#1}{#2pt}%
  \fontfamily{#3}\fontseries{#4}\fontshape{#5}%
  \selectfont}%
\fi\endgroup%
\begin{picture}(4430,1581)(139,-563)
\put(4480,424){\makebox(0,0)[lb]{\smash{{\SetFigFont{12}{14.4}{\rmdefault}{\mddefault}{\updefault}$z$}}}}
\put(301,464){\makebox(0,0)[lb]{\smash{{\SetFigFont{12}{14.4}{\rmdefault}{\mddefault}{\updefault}$Q$}}}}
\put(2851,464){\makebox(0,0)[lb]{\smash{{\SetFigFont{12}{14.4}{\rmdefault}{\mddefault}{\updefault}$Q$}}}}
\put(920,404){\makebox(0,0)[lb]{\smash{{\SetFigFont{12}{14.4}{\rmdefault}{\mddefault}{\updefault}$x$}}}}
\put(3487,397){\makebox(0,0)[lb]{\smash{{\SetFigFont{12}{14.4}{\rmdefault}{\mddefault}{\updefault}$x$}}}}
\put(3989,418){\makebox(0,0)[lb]{\smash{{\SetFigFont{12}{14.4}{\rmdefault}{\mddefault}{\updefault}$y$}}}}
\end{picture}%

%% file: Rule-two-three.pdf_t
\begin{picture}(0,0)%
\includegraphics{Rule-two-three.pdf}%
\end{picture}%
\setlength{\unitlength}{3947sp}%
\begingroup\makeatletter\ifx\SetFigFont\undefined%
\gdef\SetFigFont#1#2#3#4#5{%
  \reset@font\fontsize{#1}{#2pt}%
  \fontfamily{#3}\fontseries{#4}\fontshape{#5}%
  \selectfont}%
\fi\endgroup%
\begin{picture}(7233,1742)(-2486,-798)
\put(2026,464){\makebox(0,0)[lb]{\smash{{\SetFigFont{12}{14.4}{\rmdefault}{\mddefault}{\updefault}{\color[rgb]{0,0,0}Rule 3}%
}}}}
\put(-2324, 89){\makebox(0,0)[lb]{\smash{{\SetFigFont{12}{14.4}{\rmdefault}{\mddefault}{\updefault}$Q$}}}}
\put(-1855,122){\makebox(0,0)[lb]{\smash{{\SetFigFont{12}{14.4}{\rmdefault}{\mddefault}{\updefault}$v_0$}}}}
\put(-1877,-229){\makebox(0,0)[lb]{\smash{{\SetFigFont{12}{14.4}{\rmdefault}{\mddefault}{\updefault}$v_3$}}}}
\put(-1231,298){\makebox(0,0)[lb]{\smash{{\SetFigFont{12}{14.4}{\rmdefault}{\mddefault}{\updefault}$v_1$}}}}
\put(-1259,-543){\makebox(0,0)[lb]{\smash{{\SetFigFont{12}{14.4}{\rmdefault}{\mddefault}{\updefault}$v_2$}}}}
\put(2851, 89){\makebox(0,0)[lb]{\smash{{\SetFigFont{12}{14.4}{\rmdefault}{\mddefault}{\updefault}$Q$}}}}
\put(3320,122){\makebox(0,0)[lb]{\smash{{\SetFigFont{12}{14.4}{\rmdefault}{\mddefault}{\updefault}$v_0$}}}}
\put(3298,-229){\makebox(0,0)[lb]{\smash{{\SetFigFont{12}{14.4}{\rmdefault}{\mddefault}{\updefault}$v_3$}}}}
\put(3944,298){\makebox(0,0)[lb]{\smash{{\SetFigFont{12}{14.4}{\rmdefault}{\mddefault}{\updefault}$v_1$}}}}
\put(3916,-543){\makebox(0,0)[lb]{\smash{{\SetFigFont{12}{14.4}{\rmdefault}{\mddefault}{\updefault}$v_2$}}}}
\put(301,464){\makebox(0,0)[lb]{\smash{{\SetFigFont{12}{14.4}{\rmdefault}{\mddefault}{\updefault}$Q$}}}}
\put(767,497){\makebox(0,0)[lb]{\smash{{\SetFigFont{12}{14.4}{\rmdefault}{\mddefault}{\updefault}$v_0$}}}}
\put(1329,-141){\makebox(0,0)[lb]{\smash{{\SetFigFont{12}{14.4}{\rmdefault}{\mddefault}{\updefault}$v_2$}}}}
\put(764,139){\makebox(0,0)[lb]{\smash{{\SetFigFont{12}{14.4}{\rmdefault}{\mddefault}{\updefault}$v_3$}}}}
\put(1357,646){\makebox(0,0)[lb]{\smash{{\SetFigFont{12}{14.4}{\rmdefault}{\mddefault}{\updefault}$v_1$}}}}
\put(-524,464){\makebox(0,0)[lb]{\smash{{\SetFigFont{12}{14.4}{\rmdefault}{\mddefault}{\updefault}{\color[rgb]{0,0,0}Rule 2}%
}}}}
\end{picture}%

%% file: Rule-four.pdf_t
\begin{picture}(0,0)%
\includegraphics{Rule-four.pdf}%
\end{picture}%
\setlength{\unitlength}{3947sp}%
\begingroup\makeatletter\ifx\SetFigFont\undefined%
\gdef\SetFigFont#1#2#3#4#5{%
  \reset@font\fontsize{#1}{#2pt}%
  \fontfamily{#3}\fontseries{#4}\fontshape{#5}%
  \selectfont}%
\fi\endgroup%
\begin{picture}(5087,1829)(139,-805)
\put(5211,-727){\makebox(0,0)[lb]{\smash{{\SetFigFont{12}{14.4}{\rmdefault}{\mddefault}{\updefault}$y_3$}}}}
\put(301,464){\makebox(0,0)[lb]{\smash{{\SetFigFont{12}{14.4}{\rmdefault}{\mddefault}{\updefault}$Q$}}}}
\put(3301,464){\makebox(0,0)[lb]{\smash{{\SetFigFont{12}{14.4}{\rmdefault}{\mddefault}{\updefault}$Q$}}}}
\put(1277,652){\makebox(0,0)[lb]{\smash{{\SetFigFont{12}{14.4}{\rmdefault}{\mddefault}{\updefault}$v_1$}}}}
\put(1582,205){\makebox(0,0)[lb]{\smash{{\SetFigFont{12}{14.4}{\rmdefault}{\mddefault}{\updefault}$v_2$}}}}
\put(1584,-87){\makebox(0,0)[lb]{\smash{{\SetFigFont{12}{14.4}{\rmdefault}{\mddefault}{\updefault}$v_3$}}}}
\put(1296,-273){\makebox(0,0)[lb]{\smash{{\SetFigFont{12}{14.4}{\rmdefault}{\mddefault}{\updefault}$v_4$}}}}
\put(850,-260){\makebox(0,0)[lb]{\smash{{\SetFigFont{12}{14.4}{\rmdefault}{\mddefault}{\updefault}$v_5$}}}}
\put(860,318){\makebox(0,0)[lb]{\smash{{\SetFigFont{12}{14.4}{\rmdefault}{\mddefault}{\updefault}$v_0$}}}}
\put(3900,312){\makebox(0,0)[lb]{\smash{{\SetFigFont{12}{14.4}{\rmdefault}{\mddefault}{\updefault}$v_0$}}}}
\put(3916,-247){\makebox(0,0)[lb]{\smash{{\SetFigFont{12}{14.4}{\rmdefault}{\mddefault}{\updefault}$v_5$}}}}
\put(4491,639){\makebox(0,0)[lb]{\smash{{\SetFigFont{12}{14.4}{\rmdefault}{\mddefault}{\updefault}$u_1$}}}}
\put(4915, 26){\makebox(0,0)[lb]{\smash{{\SetFigFont{12}{14.4}{\rmdefault}{\mddefault}{\updefault}$u_2$}}}}
\put(4764,-367){\makebox(0,0)[lb]{\smash{{\SetFigFont{12}{14.4}{\rmdefault}{\mddefault}{\updefault}$u_3$}}}}
\put(5204,792){\makebox(0,0)[lb]{\smash{{\SetFigFont{12}{14.4}{\rmdefault}{\mddefault}{\updefault}$y_1$}}}}
\end{picture}%